\documentclass[12pt]{amsart}

\usepackage{amsmath, amsthm, amssymb,amscd}
\usepackage{fullpage, paralist,enumerate}
\usepackage{verbatim}
\usepackage[all]{xy}
\usepackage[parfill]{parskip}
\usepackage{graphicx}

\usepackage{afterpage}

\usepackage{amsmath}
\usepackage{amsfonts}
\usepackage{mathcomp}
\usepackage{textcomp}
\usepackage{amssymb}
\usepackage{latexsym}
\usepackage{graphicx}
\usepackage{courier}
\usepackage[english]{babel}
\usepackage{mathbbol}
\usepackage{multirow}
\usepackage{latexsym}
\usepackage{epsfig}
\usepackage{subfigure}
\usepackage{dcolumn}
\usepackage{bm}
\usepackage{colordvi}
\usepackage{color}
\usepackage{booktabs}
\usepackage{stmaryrd}
\usepackage{cite}
\usepackage[linesnumbered,commentsnumbered,ruled]{algorithm2e}
\usepackage[noend]{algpseudocode}
\usepackage{epstopdf}
\input xy
\xyoption{all}
\newcommand{\commentout}[1]{}

\newtheorem{thm}{Theorem}[section]
\newtheorem{theorem}[thm]{Theorem}

\newtheorem{prop}[thm]{Proposition}
\newtheorem{cor}[thm]{Corollary}
\newtheorem{ex}[thm]{Example}
\newtheorem{rmk}[thm]{Remark}

\newtheorem{lemma}[thm]{Lemma}

\newcommand{\nwc}{\newcommand*}

\nwc{\ben}{\begin{equation*}}
\nwc{\bea}{\begin{eqnarray}}
\nwc{\beq}{\begin{eqnarray}}
\nwc{\bean}{\begin{eqnarray*}}
\nwc{\beqn}{\begin{eqnarray*}}
\nwc{\beqast}{\begin{eqnarray*}}

\nwc{\eal}{\end{align}}
\nwc{\een}{\end{equation*}}
\nwc{\eea}{\end{eqnarray}}
\nwc{\eeq}{\end{eqnarray}}
\nwc{\eean}{\end{eqnarray*}}
\nwc{\eeqn}{\end{eqnarray*}}

\CompileMatrices

\theoremstyle{remark}

\nwc{\nn}{\nonumber}
\nwc{\mb}{\mathbf}
\nwc{\ml}{\mathcal}

\newcommand{\lt}{\left}
\newcommand{\rt}{\right}

\nwc{\vep}{\varepsilon}
\nwc{\ep}{\epsilon}
\nwc{\vrho}{\varrho}
\nwc{\orho}{\bar\varrho}
\nwc{\vpsi}{\varpsi}
\nwc{\lamb}{\lambda}
\nwc{\om}{\omega}
\nwc{\Om}{\Omega}
\nwc{\al}{\alpha}
\nwc{\sgn}{\mbox{\rm sgn}}

\nwc{\IA}{\mathbb{A}} 
\nwc{\bi}{\mathbf{i}}
\nwc{\ba}{\mathbf{a}}
\nwc{\bmb}{\mathbf{b}}
\nwc{\bo}{\mathbf{o}}
\nwc{\IB}{\mathbb{B}}
\nwc{\IC}{\mathbb{C}} 
\nwc{\ID}{\mathbb{D}} 
\nwc{\IM}{\mathbb{M}} 
\nwc{\IP}{\mathbb{P}} 
\nwc{\II}{\mathbb{I}} 
\nwc{\IE}{\mathbb{E}} 
\nwc{\IF}{\mathbb{F}} 
\nwc{\IG}{\mathbb{G}} 
\nwc{\IN}{\mathbb{N}} 
\nwc{\IQ}{\mathbb{Q}} 
\nwc{\IR}{\mathbb{R}} 
\nwc{\IT}{\mathbb{T}} 
\nwc{\IZ}{\mathbb{Z}} 

\nwc{\cE}{{\ml E}}
\nwc{\cP}{{\ml P}}
\nwc{\cQ}{{\ml Q}}
\nwc{\cL}{{\ml L}}
\nwc{\cX}{{\ml X}}
\nwc{\cW}{{\ml W}}
\nwc{\cZ}{{\ml Z}}
\nwc{\cR}{{\ml R}}
\nwc{\cV}{{\ml V}}
\nwc{\cT}{{\ml T}}
\nwc{\crV}{{\ml L}_{(\delta,\rho)}}
\nwc{\cC}{{\ml C}}
\nwc{\cO}{{\ml O}}
\nwc{\cA}{{\ml A}}
\nwc{\cK}{{\ml K}}
\nwc{\cB}{{\ml B}}
\nwc{\cD}{{\ml D}}
\nwc{\cF}{{\ml F}}
\nwc{\cS}{{\ml S}}
\nwc{\cM}{{\ml M}}
\nwc{\cG}{{\ml G}}
\nwc{\cH}{{\ml H}}
\nwc{\bk}{{\mb k}}
\nwc{\bn}{{\mb n}}
\nwc{\bp}{{\mb p}}
\nwc{\bz}{\mb z}
\nwc{\bl}{{\mb l}}
\nwc{\bj}{{\mb j}}
\nwc{\bs}{{\mb s}}
\nwc{\by}{\mathbf{h}}
\nwc{\bZ}{\mathbf{Z}}
\nwc{\bF}{\mathbf{F}}
\nwc{\bE}{\mathbf{E}}
\nwc{\bV}{\mathbf{V}}
\nwc{\bY}{\mathbf Y}
\nwc{\br}{\mb r}
\nwc{\pft}{\cF^{-1}_2}
\nwc{\bU}{{\mb U}}
\nwc{\bG}{{\mb G}}
\nwc{\bg}{\mathbf{g}}
\nwc{\mbf}{\mathbf{f}}
\nwc{\mbe}{\mathbf{e}}
\nwc{\be}{\mathbf{e}}
\nwc{\ind}{\operatorname{I}}
\nwc{\mbx}{\mathbf{f}}
\nwc{\bb}{\mathbf{g}}
\nwc{\xmax}{f_{\rm max}}
\nwc{\xmin}{f_{\rm min}}
\nwc{\suppx}{\hbox{\rm supp} (\mbf)}
\nwc{\cI}{\IZ^2_N}
\nwc{\chis}{{\chi^{\rm s}}}
\nwc{\chii}{{\chi^{\rm i}}}
\nwc{\pdfi}{{f^{\rm i}}}
\nwc{\pdfs}{{f^{\rm s}}}
\nwc{\pdfii}{{f_1^{\rm i}}}
\nwc{\pdfsi}{{f_1^{\rm s}}}
\nwc{\thetatil}{{\tilde\theta}}
\nwc{\red}{\color{red}}
\nwc{\blue}{\color{blue}}

\nwc{\prox}{\hbox{prox}}
\nwc{\diag}{\hbox{\rm diag}}
\nwc{\supp}{{\hbox{\rm supp}}}

\nwc{\sloc}{J_{\rm f}}
\nwc{\bu}{{\mb u}}
\nwc{\bv}{{\mb v}}
\nwc{\cU}{\mathcal{U}}
\nwc{\cN}{\mathcal{N}}
\nwc{\bN}{\mathbf{N}}
\nwc{\mbm}{\mathbf{m}}
\nwc{\bw}{\mathbf{w}}
\nwc{\bom}{\mathbf{w}}
\nwc{\bt}{\mathbf{t}}
\nwc{\z}{y}
\nwc{\cY}{\mathcal{Y}}
\nwc{\bM}{\mathbf{M}}
\nwc{\half}{{1\over 2}}
\nwc{\Sf}{S_{\rm f}}
\nwc{\Jf}{J_{\rm f}}
\nwc{\nul}{\hbox{\rm null}_\IR}
\nwc{\spanR}{\hbox{\rm span}_\IR}
\nwc{\Arg}{\hbox{\rm Arg~}}
\nwc{\fdr}{S_{\rm f}}
\nwc{\phase}[1]{\exp\lt[i\measured #1\rt]}

\nwc{\im}{{\rm i}}

\nwc{\lb}{\llbracket}
\nwc{\rb}{\rrbracket}

\begin{document}
 \title{
Raster Grid Pathology and the Cure 
}
 
 \author{Albert Fannjiang 
 \address{
Department of Mathematics, University of California, Davis, California  95616, USA. Email:  {\tt fannjiang@math.ucdavis.edu}
}}

  \date{}

\maketitle 

\begin{abstract} Blind ptychography is a phase retrieval method using multiple coded diffraction patterns from different, overlapping parts of the unknown extended  object illuminated with an unknown  window function.  The window function is also known as 
the probe in the optics literature.  As such blind ptychography is an inverse problem of  simultaneous recovery of the object and the window function given the intensities of the windowed Fourier transform and  has a multi-scale set-up in which the probe has an intermediate scale between the pixel scale and the macro-scale of the extended object. 
Uniqueness problem for blind ptychography is analyzed rigorously  for the raster scan (of a constant step size $\tau$) and its variants, in which
another scale comes into play: the overlap between adjacent blocks
(the shifted windows). 
The block phases are shown to form an arithmetic progression
and the complete characterization of the raster scan ambiguities is given,
including: First, the periodic raster grid pathology
of  degrees of freedom  proportional to $\tau^2$ and, second, a non-periodic,  arithmetically progressing phase shift from block to block.
Finally irregularly perturbed raster scans are shown  to remove all ambiguities other than the inherent ambiguities of the scaling factor and
the affine phase ambiguity under the minimum requirement of roughly $50\%$ overlap ratio. \end{abstract}


\section{Introduction}

In the last decade, ptychography has made rapid technological advances and developed into a powerful lensless coherent imaging method \cite{DM08,EM-ptych2,FPM13}. 
Ptychography  collects the diffraction patterns from overlapping illuminations of various parts of the unknown object 
using  a localized coherent source (the probe) \cite{Pfeiffer,Nugent,Rod08}, and  builds on the advances in synthetic aperture methods to extend phase retrieval 
 to unlimited objects and enhance imaging resolution \cite{random-aperture, EM0,EM1,supres-PIE,ptycho10}.
Blind ptychography goes a step further and seeks to reconstruct both the unknown object and the unknown probe simultaneously \cite{probe09,Yang14}.

Mathematically, blind ptychography is an inverse problem of  simultaneous recovery of the object and the window function (the probe) given the intensities of the windowed Fourier transform. In ptychography, the window function has an intermediate scale between the pixel scale and the macro-scale of the extended object. 
 
 The performance of ptychography depends on factors such as the type of illumination
 and the measurement scheme, including the amounts of overlap
 and probe positions. 
 For example, the use of randomly structured illuminations can improve ptychographic reconstruction over that
 with  regular illuminations \cite{Fucai2,Spread,random-coding,diffuser,RCM,random-aperture,ptycho-rpi,Horisaki1,Horisaki2, DRS-ptych,ptych-unique,rpi,FAK}. 
 Experiments suggest an  overlap ratio of at least 50\%, typically 
60-70\% between adjacent illuminations for blind ptychography \cite{overlap,ePIE09}. Optimizing the scan pattern can significantly improve
the performance of ptychography and 
is an important part of the  experimental design.

In particular, empirical evidences  repeatedly point to the pitfalls of the raster scan, which is experimentally the easiest to implement \cite{artifact}. Mathematically speaking, blind ptychography with raster scan seeks to recover  both the object and the window function
(the probe) as unknowns
but only the 2D windowed Fourier {\em intensities}  (coded diffraction patterns)
as the data.  Raster scanning  refers to the positions of the window function. 
The raster scan scheme  is susceptible to periodic artifacts, known as {\em raster grid pathology}, attributed  to the regularity and symmetry of the scan positions  \cite{probe09}. 

On the other hand, to the best of our knowledge, raster grid pathology has not be precisely formulated and analyzed. 
The purpose of the present work is a complete analysis of raster grid pathology from the perspective of  inverse problems. 
Uniqueness of solution is fundamental to any inverse problem. 
The exceptions to uniqueness are called the ambiguities. We identify the rater grid pathology reported in optics literature  as 
{\em periodic} ambiguities of period equal to the step size of the raster scan. Moreover, we will characterize all the ambiguities
inherent to the raster scan ptychography and propose a simple modification that can eliminate all the ambiguities
except for those inherent to {\em any} blind ptychography.

The first thing to note is that raster grid pathology only appears in blind ptychography but not in ptychography with
a known probe. In the latter case, the only ambiguity is a constant phase factor which has no real significance (and will be ignored)
and the convergence behaviors of the raster scan ptychography with a known probe has been rigorously established 
\cite{ptych-unique}. 

Second, there are two ambiguities inherent to any blind ptychography: a scaling factor and
an affine phase factor. To give a precise description, we introduce some notation as follows. 

Let $\IZ_n^2=\lb 0,n-1\rb^2$ be the object domain containing the support of the discrete object $f$ where $\lb k, l\rb$ denotes the integers between, and including,  $k\le l\in \IZ$.  Let $\cM^{00}:=\IZ_m^2, m<n,$ be the initial probe area which is  also the support of the probe $\mu^{00}$ describing the illumination field.  Here $n$ is the global scale and $m$ the intermediate scale of the set-up. 

Let $\cT$ be the set of all shifts, including $(0,0)$,  involved  in the ptychographic measurement. 
 Denote by $\mu^\bt$ the $\bt$-shifted probe for all $\bt\in \cT$ and $\cM^\bt$ the domain of
$\mu^\bt$. Let $f^\bt$ the object restricted to $\cM^\bt$.
\commentout{and  $\mb{\rm Twin}(f^\bt)$ the twin image of $f^\bt$ in $\cM^\bt$ defined as
\[
\mb{\rm Twin}(f^\bt)(\bn)=\bar f^\bt(\bN+2\bt-\bn),\quad\bn\in \cM^\bt,\quad\bN=(n,n).
\]
}
We write $f=\vee_\bt f^\bt$ and refer to each $f^\bt$ as a part of $f$. In ptychography, the original object is broken up into a set of overlapping object parts, each of which produces a $\mu^\bt$-coded diffraction pattern  (i.e. Fourier intensity).  
The totality of the coded diffraction patterns is called the ptychographic measurement data.  Let $\nu^{00}$ (with $\bt=(0,0)$) and $g=\vee_\bt g^\bt$ be any pair
 of the probe and the object estimates producing  the same ptychography data as $\mu^{00}$ and $f$, i.e.
 the diffraction pattern of $\nu^\bt\odot g^\bt$ is identical to that of $\mu^\bt\odot f^\bt$ where
 $\nu^\bt$ is the $\bt$-shift of $\nu^{00}$ and $g^\bt$ is the restriction of $g$ to $\cM^\bt$. 
 For convenience, we assume the value zero for $\mu^\bt, f^\bt,\nu^\bt, g^\bt $ outside of $\cM^\bt$
 and the periodic boundary condition on $\IZ_n^2$ when $\mu^\bt$ crosses over the boundary of $\IZ_n^2$.

Consider the probe and object estimates
\beq
\label{lp1}
\nu^{00}(\bn)&=&\mu^{00}(\bn) \exp(-\im a -\im \bw\cdot\bn),\quad\bn\in\cM^{00}\\
\label{lp2} g(\bn)&=& f(\bn) \exp(\im b+\im \bw\cdot\bn),\quad\bn\in \IZ^2_n
\eeq
for any $a,b\in \IR$ and $\bw\in \IR^2$.  For any $\bt$, we have the following
calculation
\beqn
\nu^\bt(\bn)&=&\nu^{00}(\bn-\bt)\\
&=&\mu^{00}(\bn-\bt) \exp(-\im \bw\cdot(\bn-\bt))\exp(-\im a)\\
&=&\mu^\bt(\bn) \exp(-\im \bw\cdot(\bn-\bt))\exp(-\im a)
\eeqn
and hence for all $\bn\in \cM^\bt, \bt\in\cT$
\beq
\label{drift2}
\nu^\bt(\bn) g^\bt(\bn)&=&\mu^\bt(\bn)f^\bt(\bn) \exp(\im(b-a))\exp(\im \bw\cdot\bt). 
\eeq
Clearly, \eqref{drift2} implies that $g$ and $\nu^{00}$ produce the same ptychographic data as $f$ and $\mu^{00}$ since
for each $\bt$, $\nu^\bt\odot g^\bt$ is a constant phase factor times $\mu^\bt\odot f^\bt$. 

In addition to the affine phase ambiguity \eqref{lp1}-\eqref{lp2}, another ambiguity,  a scaling factor ($g=c f, \nu^{00}=c^{-1} \mu^{00}, c>0$),
is also inherent to any blind ptychography as can easily be checked. We refer to 
 the scaling factor and the affine phase ambiguity as the inherent ambiguities of blind ptychography. 
Note that when the probe is exactly known $\nu^{00}=\mu^{00}$, neither ambiguity can occur.

A recent theory of uniqueness for blind ptychography with random probes  \cite{blind-ptych} establishes  that  for general sampling schemes and with high probability  (in  the selection of the random probe), we have the relation
\beq
\label{7.1}
\nu^{\bt}\odot g^\bt&=&e^{\im\theta_{\bt}}\mu^{\bt}\odot f^\bt,\quad \bt\in \cT, 
\eeq
for some constants $\theta_\bt\in \IR$ (called block phases here)
 if $g$ and $\nu^\bt$ produce the same diffraction pattern as $f$ and $\mu^\bt$ for all $\bt\in \cT$. 
  Here $\odot$ denotes the component-wise (Hadamard) product. The masked object parts $\psi^\bt:=\mu^\bt\odot f^\bt$ are also known
as the {\em exit waves} in the scanning transmission electron microscopy literature. 

We refer to  \eqref{7.1} as the {\em local uniqueness} of the exit waves which means unique determination of  the exit waves 
 up to the block phases but not globally since $\theta_\bt$
 can depend on $\bt$ and vary from block to block. However, the block phase profile is  not arbitrary. For example, block phases for the raster scan and the perturbed raster scan always form an arithmetic progression (see below), possessing two degrees of freedom. 
 
 Once the exit waves $\psi^\bt$ are determined up to block phases,
 \eqref{7.1} with $\theta_\bt$ treated as parameters represents a bilinear system (in $\nu^{00}$ and $g$) of $m^2\times |\cT|$ equations coupled through the overlap between adjacent blocks. The total number of complex variables is $n^2+m^2$. In the case of raster scan with step size $\tau$,  $|\cT|\approx {n^2/\tau^2}$ and $m^2|\cT|\approx n^2(\tau/m)^{-2}$ where  the shift ratio $\tau/m$ is 1 minus the overlap ratio
 $(m-\tau)/m$. For $50\%$ overlap ratio and $m<n$, $m^2|\cT|\approx 4n^2$,  a couple times larger than $(n^2+m^2)$. This speaks of the potential  redundancy of information in \eqref{7.1} on dimension count. Yet this simplistic analysis is deceptive as we will see that due to degenerate coupling the raster scan has ambiguities of exactly $\tau^2+2$ degrees of freedom in addition to the three degrees of freedom of the inherent ambiguities discussed above. 

\commentout{
Since both the scaling factor and the affine phase ambiguity are inherent to any blind ptychography, this means that
an affine profile are part of the intrinsic degrees of freedom of the block phases. 
However, the affine phase ambiguity is not the only cause of  the affine profile in the block phases 
for blind raster scan ptychography. 
}

 We will take \eqref{7.1} as the starting point of our analysis of raster scan ambiguities, first to characterize 
 all the ambiguities in the raster scan and, second, to show how to harness 
 the nonlinear coupling in  \eqref{7.1} by more nuanced design of measurement schemes in which pixel-scale changes result in total eradication of ambiguities other than the inherent ones  through  the intermediate-scale coupling.

 \subsection{Our contribution}
We first prove that  the block phases of the raster scan of any step size $\tau<m$ always have
an affine profile (Section \ref{sec:drift}, Theorem \ref{lem2}). 
We then give a complete characterization of the  raster scan ambiguities (Theorem \ref{thm:recovery}).

Roughly speaking, there are two  types of ambiguities besides the inherent ambiguities (the scaling factor and the affine phase ambiguity \eqref{lp1}-\eqref{lp2}). 
First, there is the non-periodic, arithmetically progressing ambiguity,  inherited from the aforementioned affine block phase profile,
which varies  on the block scale while the affine phase ambiguity varies on the pixel scale. 

Second, there are $\tau$-periodic ambiguities of $\tau^2$ degrees of freedom, which we identify as
mathematical description of  the raster grid pathology reported in the optics literature.
The larger the step size the (much) greater the degrees of ambiguity which can not be removed without
extra prior information. 

Finally we demonstrate a simple mechanism for eliminating all the other ambiguities than the scaling  factor and the affine phase ambiguity by slightly perturbing
the raster scan with the minimum overlap ratio roughly $50\%$, consistent with experimental findings in the optics literature   (Section \ref{sec:mix}, Theorem \ref{thm:mix}). 
The optimal tradeoff between the speed of data acquisition and the convergence rate of reconstruction 
lies in the balance between the average step size and the overlap size. 


The rest of the paper is organized as follows. 
In Section \ref{sec:lattice}, we give a detailed presentation of the raster scan. In Section \ref{sec:drift}, we prove
that the block phases have an affine profile. In Section \ref{sec:recovery}, we give a complete characterization of
the raster scan ambiguities. In Section \ref{sec:mix} we show that slightly perturbed raster scan has no other ambiguities than
the scaling factor and the affine phase ambiguity. In Section \ref{sec:num}, we give numerical demonstrate of  the perturbed raster scan. We conclude with a few remarks in Section \ref{sec:con}. 

  \begin{figure}[t]
\begin{center}
\subfigure[Ptychography set-up]{\includegraphics[width=7cm]{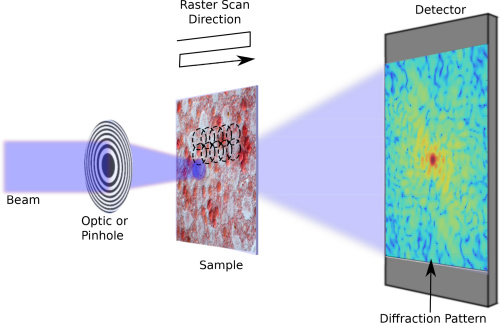}}\hspace{2cm}
\subfigure[raster scan pattern]{\includegraphics[width=5cm]{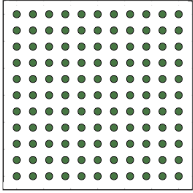}}
\caption{Simplified ptychographic setup showing a Cartesian grid used for the overlapping raster scan positions \cite{parallel}.
}
\label{fig0}
\end{center}
\end{figure}

\section{Raster scan}
\label{sec:lattice}
\commentout{
A lattice in $ \mathbb{R}^{n}$ is a subgroup of $\IR^n$ which is isomorphic to $ \mathbb{Z}^{n}$. In other words, for any basis of $\mathbb{R}^{n}$, the subgroup of all linear combinations with integer coefficients of the basis vectors forms a lattice. A lattice may be viewed as a regular tiling of a space by a primitive cell.

A lattice is the symmetry group of discrete translational symmetry in $n $ directions.  As a group (dropping its geometric structure) a lattice is a finitely-generated free abelian group, and thus isomorphic to $\mathbb{Z}^{n}$.

A wallpaper group (or plane symmetry group or plane crystallographic group) is a two-dimensional symmetry group. Mathematically, a wallpaper group or plane crystallographic group is a type of topologically discrete group of isometries of the Euclidean plane that contains two linearly independent translations.
}

The raster scan can be formulated as the 2D lattice with the basis $\{\bv_1, \bv_2\}$
\beq\label{lattice}
\cT=\{\bt_{kl}\equiv k\bv_1+l\bv_2: k,l\in \IZ\},\quad \bv_1,\bv_2\in \IZ^2
\eeq
acting on the object domain $\IZ_n^2$. 
Instead of $\bv_1$ and $\bv_2$ we can also take $ \bu_1=\ell_{11}\bv_1+\ell_{12}\bv_2$ and $\bu_2=\ell_{21} \bv_1+\ell_{22}\bv_2$ for integers $\ell_{ij}$ with $\ell_{11}\ell_{22}-\ell_{12}\ell_{21}=\pm 1$. This ensures that $\bv_1$ and $\bv_2$ themselves are integer linear combinations of $\bu_1,\bu_2$. Every lattice basis defines a fundamental parallelogram, which determines the lattice. 
There are five 2D lattice types, called period lattices, as given by the crystallographic restriction theorem. In contrast,
there are 14 lattice types in 3D, called Bravais lattices \cite{Conway}. 

We will focus on the simplest raster scan corresponding to  the  {\em square lattice} with $\bv_1=(\tau,0), \bv_2=(0,\tau)$ of  step size $\tau\in \IN$. Our results can easily be extended to
other lattice schemes.

Under  the periodic boundary condition the raster scan with the step size $\tau=n/q, q\in \IN,$ $\cT$ consists of $\bt_{kl}={\tau}(k,l)$, with $k,l\in \{0,1,\cdots, q-1\}$. The periodic boundary condition means that for $k=q-1$ or $l=q-1$  the shifted probe is wrapped around into the other end of the object domain.  
  Denote the $\bt_{kl}$-shifted  probes and blocks by $\mu^{kl}$ and $\cM^{kl}$, respectively. Likewise, denote
by $f^{kl}$ the object restricted to the shifted domain $\cM^{kl}$.

Depending on whether $\tau\le m/2$ (the under-shifting case) or $\tau>m/2$ (the over-shifting case), we have two types of schemes. For the former case,  all pixels of the the object participate
in an equal number of diffraction patterns. For the latter case, however, $4(m-\tau)^2$ pixels participate in four,
$4(2{\tau}-m)(m-{\tau})$ pixels participate in two 
and $(2{\tau}-m)^2$  pixels participate in only one diffraction pattern, resulting in uneven coverage of the object.

\subsection{The under-shifting scheme $\tau\le m/2$}
\label{sec:under}
For simplicity of presentation we consider the case of $\tau=m/p$ for some integer $p\ge 2$ (i.e. $pn=qm$). 
As noted above, all pixels of the the object participate
in the same number (i.e. $2p$) of diffraction patterns. 
The borderline case $\tau=m/2$   (dubbed the minimalist scheme in \cite{ptych-unique}) corresponds to
$p=2$.

We
partition the cyclical $\bt^{kl}$-shifted probe $\mu^{kl}$ and the corresponding domain 
into equal-sized square blocks as 
\beq
\label{mu}
\mu^{kl}&=&
\left[
\begin{matrix}
\mu^{kl}_{00}  & \mu^{kl}_{10} &\cdots&  \mu^{kl}_{p-1,0}   \\
\mu^{kl}_{01}&\mu^{kl}_{11}&\cdots& \mu^{kl}_{p-1,1}\\
\vdots&\vdots&\vdots&\vdots\\
\mu^{kl}_{0,p-1}  &\mu^{kl}_{1,p-1}&\cdots& \mu^{kl}_{p-1,p-1} 
\end{matrix}
\right],\quad \mu^{kl}_{ij}\in \IC^{m/p\times m/p}\\
 \cM^{kl}&=&
\left[
\begin{matrix}
\cM^{kl}_{00}  & \cM^{kl}_{10} &\cdots&  \cM^{kl}_{p-1,0}   \\
\cM^{kl}_{01}&\cM^{kl}_{11}&\cdots& \cM^{kl}_{p-1,1}\\
\vdots&\vdots&\vdots&\vdots\\
\cM^{kl}_{0,p-1}  &\cM^{kl}_{1,p-1}&\cdots& \cM^{kl}_{p-1,p-1} 
\end{matrix}
\right],\quad \cM^{kl}_{ij}\in \IZ^{m/p\times m/p}
\eeq
under the periodic boundary condition 
\beq
\label{bc1''}
\mu^{q-1-i,k}_{j,l}=\mu^{0k}_{j-i-1, l}, &&
\mu^{k,q-1-i}_{l,j}=\mu^{k0}_{l,j-i-1}, \\
\cM^{q-1-i,k}_{j,l}=\cM^{0k}_{j-i-1, l}, &&
\cM^{k,q-1-i}_{l,j}=\cM^{k0}_{l,j-i-1}
\eeq
for all $ 0\le i\le j-1\le p-2, k=1,\dots,q-1, l=1,\dots, p-1.$

Accordingly, we divide the object $f$ into $q^2$ non-overlapping square blocks
  \beq \label{fp}
 f=
\left[
\begin{array}{ccc}
f_{00}  & \ldots  & f_{q-1,0}   \\
\vdots  & \vdots  &\vdots   \\
f_{0,q-1}  &\ldots   &  f_{q-1, q-1}  
\end{array}
\right],\quad f_{ij}\in \IC^{m/p\times m/p}. \eeq

  \commentout{
Let $\Phi$ denote the 2D oversampled Fourier transforms divided  into $p\times p$ blocks
\[ \Phi=
\left[
\begin{matrix}
\Phi_{00}  &  \Phi_{10} &\cdots& \Phi_{p-1,0}\\
\Phi_{01} &\Phi_{11}&\cdots& \Phi_{p-1,1}\\
\vdots&\vdots&\vdots&\vdots\\
\Phi_{0, p-1}&\Phi_{1,p-1}&\cdots& \Phi_{p-1,p-1}
\end{matrix}
\right]
\]
where  $\Phi_{ij}:  \cM^{kl}_{ij} \to \IC^{2m/p \times 2m/p}$ is normalized such that
     \beq 
\Phi^*_{ij} \Phi_{i'j'}=p^{-2}{\delta_{ij,i'j'}}I_{m/p\times m/p}, \quad \forall i,j,i',j'.
\eeq
The diffracted field  $H$ can be partitioned into $q\times q$ blocks, $[H_{ij}]$,
where
\[ 
H_{ij}=\sum_{k,l=0}^{p-1}
\Phi_{kl}(\mu^{i-1,j-1}_{kl}\odot f^{i+k,j+l}), \quad i,j =0,1, \ldots, q-1
\]
where $f_{i+k,j+l}$ is cyclically defined with respect to the subscript. 
}

 \subsection{The over-shifting scheme $\tau>m/2$.} \label{sec:over}
 Because of uneven coverage of the object domain, the over-shifting case is more complicated. 
 
We divide the shifted probe $\mu^{kl}$ and its domain as
\beq\label{mu2}
\mu^{kl}&=&
\left[
\begin{matrix}
\mu^{kl}_{00}  & \mu^{kl}_{10} &  \mu^{kl}_{20}   \\
\mu^{kl}_{01}&\mu^{kl}_{11}&\mu^{kl}_{21}\\
\mu^{kl}_{02}  &\mu^{kl}_{12}& \mu^{kl}_{22} 
\end{matrix}
\right]\in\IC^{m\times m}\\
 \cM^{kl}&=&
\left[
\begin{matrix}
\cM^{kl}_{00}  & \cM^{kl}_{10} &  \cM^{kl}_{20}   \\
\cM^{kl}_{01}&\cM^{kl}_{11}&\cM^{kl}_{21}\\
\cM^{kl}_{02}  &\cM^{kl}_{12}& \cM^{kl}_{22} 
\end{matrix}\right]\in \IZ^{m\times m}
\eeq
under the periodic boundary condition 
\beq
\label{bc1'}
\cM^{q-1,k}_{2j}=&\cM^{0k}_{0j},\quad 
\cM^{k,q-1}_{i2}=&\cM^{k0}_{i0}\\
\mu^{q-1,k}_{2j}=& \mu^{0k}_{0j},
\quad \mu^{k,q-1}_{i2}= &\mu^{k0}_{i0}, \label{bc2'}
\eeq
for all $k=1,\dots,q-1$ and $i,j=0,1,2$, where $q$ is the number of shifts in each direction. 

Note that the sizes of these blocks are not equal: the four corner blocks are  $(m-\tau) \times (m-\tau) $,
the center block is $(2\tau-m) \times (2\tau -m)$ and 
the rest are either $(2\tau -m)\times (m-\tau)$  or $(m-\tau)\times (2\tau-m)$. 
As a result, the corresponding partition of $f$  also has unequally sized blocks.

We write 
  \beq \label{fp2}
f=\bigvee_{k,l=0}^{q-1} f^{kl},\quad    f^{kl}= 
\left[
\begin{matrix}
f^{kl}_{00}  & f^{kl}_{10} &  f^{kl}_{20}   \\
f^{kl}_{01}&f^{kl}_{11}&f^{kl}_{21}\\
f^{kl}_{02}  &f^{kl}_{12}& f^{kl}_{22} 
\end{matrix}
\right]\in \IC^{m\times m}
\eeq
where, for $i,j=0,1,2,\quad k,l=0,\cdots, q-1$,
\beqn
f^{kl}_{2j}=&f^{k+1,l}_{0j},\quad 
f^{k,l}_{i2}=&f^{k, l+1}_{i0}. 
\eeqn

 \commentout{
Let $\Phi$ denote the 2D oversampled Fourier transforms divided  into nine blocks corresponding to the block structure of $\cM^{00}$
\[ \Phi=
\left[
\begin{matrix}
\Phi_{00}  &  \Phi_{10} &\Phi_{20}\\
\Phi_{01} &\Phi_{11}&\Phi_{21}\\
\Phi_{02}&\Phi_{12}&\Phi_{22}
\end{matrix}
\right].
\]
The diffracted field  $H$ can be partitioned into $q\times q$ blocks, $[H_{ij}]$,
where
\[ 
H_{ij}=\sum_{k,l=0}^2
\Phi_{kl}(\mu^{i-1,j-1}_{kl}\odot f^{i+k,j+l}), \quad i,j =0,1,\ldots, q-1
\]
where $f_{i+k,j+l}$ is cyclically defined with respect to the subscript. 
}

\section{Affine block phases}\label{sec:drift}

Let $S$ be any cyclic subgroup of $\cT$ generated by $\bv$, i.e. $S:=\{\bt_j=j\bv:j=0,\dots,s-1\}$, of order $s$, i.e.  $s\bv=0\mod\, n$. 
For ease of notation, denote by $\mu^k, f^k, \nu^k, g^k$ and $M^{k}$ for the respective $\bt_k$-shifted quantities.
\medskip
\begin{theorem}\label{lem2}  As in \eqref{7.1}, suppose that
\beq
\label{25}
\nu^k(\bn)g^k(\bn)=e^{\im \theta_k}\mu^k(\bn) f^k(\bn)
\eeq
{for all $\bn\in \cM^k$ and $k=0,\dots,s-1$. }
 If, for all $ k=0,\dots,s-1,$
\beq
\label{33}
 \cM^{k}\cap \cM^{k+1}\cap \supp(f) \cap(\supp(f)+ {\bv})\neq \emptyset,
\eeq
then the sequence 
 $\{\theta_0,\theta_1,\dots,\theta_{s-1}\}$ is an arithmetic progression  where $\Delta\theta=\theta_k-\theta_{k-1}$  is an integer multiple of $2\pi/s$.
\end{theorem}
\medskip
\begin{rmk}
If $f$ has a full support, i.e. $\supp(f)=\IZ^2_n$, then \eqref{33} holds for any step size $\tau<m$ (i.e. positive overlap). 

\end{rmk}
\begin{proof}

Rewriting \eqref{25} in the form
\beq
\label{25'}
\nu^{k+1}(\bn) g^{k+1}(\bn)&=&e^{\im\theta_{k+1}}\mu^{k+1}(\bn) f^{k+1}(\bn)
\eeq
and substituting \eqref{25} into 
 \eqref{25'}  for $\bn\in \cM^k\cap \cM^{k+1}$, we have
\beqn
e^{\im\theta_{k}}f^k(\bn) \mu^{k}(\bn)/\nu^{k}(\bn)=e^{\im\theta_{k+1}}\mu^{k+1}(\bn)/\nu^{k+1}(\bn) f^{k+1}(\bn)
\eeqn
and hence for all $ \bn\in \cM^{k}\cap \cM^{k+1}\cap \supp(f)$,
\beq\label{30}
e^{\im\theta_{k}}\mu^{k}(\bn)/\nu^{k}(\bn)=e^{\im\theta_{k+1}}\mu^{k+1}(\bn)/\nu^{k+1}(\bn).
\eeq

For all $ j=0,\dots, s-1,$ substituting 
\beq
\label{31}
\nu^{j}(\bn)=\nu^{j+1}(\bn+ \bv),\quad  \mu^{j}(\bn)=\mu^{j+1}(\bn+ \bv),
\eeq
 into \eqref{30}, we have that for $\bn\in \cM^{k}\cap \cM^{k+1}\cap \supp(f)$
\beqn
\lefteqn{e^{\im \theta_{k}} \mu^{k+1}(\bn+ \bv)/\nu^{k+1}(\bn+ \bv)}\\
&=&e^{\im \theta_{k+1}}\mu^{k+2}(\bn+ \bv)/\nu^{k+2}(\bn+ \bv),\eeqn
or equivalently
\beq
\label{32}e^{\im \theta_{k}} \mu^{k+1}(\bn)/\nu^{k+1}(\bn)&=&e^{\im \theta_{k+1}}\mu^{k+2}(\bn)/\nu^{k+2}(\bn),\\
\forall  \bn&\in& \cM^{k+1}\cap \cM^{k+2}\cap (\supp(f)+ \bv)\nn
\eeq
On the other hand, \eqref{30} also implies
\beq
\label{32'}e^{\im \theta_{k+1}} \mu^{k+1}(\bn)/\nu^{k+1}(\bn)&=&e^{\im \theta_{k+2}}\mu^{k+2}(\bn)/\nu^{k+2}(\bn),\\\forall \bn &\in & \cM^{k+1}\cap \cM^{k+2}\cap \supp(f).\nn
\eeq
Hence, if 
\beqn
 \cM^{k}\cap \cM^{k+1}\cap \supp(f) \cap(\supp(f)+ \bv)\neq \emptyset
\eeqn
then \eqref{32'} and \eqref{32} imply that
\beq\label{34}
e^{\im\theta_{k+1}}e^{-\im\theta_{k}}=e^{\im\theta_{k}} e^{-\im\theta_{k-1}},\quad \forall k =0,\dots,s-1
\eeq
and hence $\Delta\theta=\theta_k-\theta_{k-1}$ is independent of $k$.
In other words, $\{\theta_0,\theta_1,\theta_2\dots\}$ is an arithmetic progression.

Moreover,  the periodic boundary condition and the fact that $s\bv=0 \mod 2\pi$ imply that
 $s\Delta \theta$ is an integer multiple of $2\pi$. \end{proof}

Applying Theorem \ref{lem2} to the two-generator group $\cT$ of the raster scan we have the following result.
\medskip
\begin{cor}\label{cor1}
For the full raster scan $\cT$,  the block phases have the profile 
\beq
\label{a3} \theta_{kl}=\theta_{00}+\br\cdot(k,l),\quad k,l=0,\dots,q-1, 
\eeq
for some $\theta_{00}\in \IR$ and $  \br=(r_1,r_2)$ where $r_1$ and $r_2$ are integer multiples
of $2\pi/q$. 

\end{cor}

\commentout{
The following examples show that in addition to the affine phase ambiguity another ambiguity is inherent to the raster scan in connection to Theorem \ref{lem2} as the arithmetically progressing block phases make positive
and negative imprints on the object and phase estimates, respectively.

\medskip
\begin{ex}\label{ex4} Let $\bv=(m/2,0), q=2n/m$ ($m$ is even). To simplify the notation for partition of the object and the probe we write 
 \beqn\label{28}
 f &=&
\lt[\begin{matrix}
f_{0}  & f_{1} & \cdots  &f_{q-1}  
\end{matrix}\rt]\\
 \mu^0&=&
\left[
\begin{matrix}
\mu_{0}  & \mu_{1}
\end{matrix}
\right]\\
 g&=&
\left[
\begin{matrix}
f_{0}  & e^{\im 2\pi/q} f_{1} & \cdots &  e^{\im 2\pi (q-1)/q} f_{q-1}  
\end{matrix}
\right]\\
\nu^0&=&
\left[
\begin{matrix}
\mu_{0}  &e^{-\im 2\pi/q}  \mu_{1} 
\end{matrix}
\right]\label{29}
\eeqn
 where $f_j, g_j\in \IC^{n\times m/2}, \mu_i, \nu_i\in \IC^{m\times m/2}$.  It is easy to see that 
\beq
\label{10.1}
\nu^k \odot g^k= e^{\im 2\pi k/q} \mu^k\odot f^k,\quad k=0,\dots,q-1,
\eeq
where for $ k=0,\dots, q-2,$
\beqn
f^k&=&[f_k,f_{k+1}],\\
g^k&=&[g_k,g_{k+1}]=[e^{\im2\pi k/q}f_k, e^{\im2\pi (k+1)/q}f_{k+1}]
\eeqn
and
\beqn
f^{q-1}&=&[f_{q-1}, f_0],\\
 g^{q-1}&=& [g_{q-1},g_0]=[e^{\im2\pi (q-1)/q}f_{q-1}, f_{0}].
\eeqn

More generally, consider $\bv=(\tau,0)$ for $\tau=m/p, q=pn/m$ where $p$ divides $m$:
 \beq\label{28'}
 f&=&
\lt[\begin{matrix}
f_{0}  & f_{1} & \cdots  &f_{q-1}  
\end{matrix}\rt]\\
 \mu^0&=&
\left[
\begin{matrix}
\mu_{0}  & \mu_{1}&\cdots &\mu_{p-1}
\end{matrix}
\right]\\
 g&=&
\left[
\begin{matrix}
f_{0}  & e^{\im 2\pi/q} f_{1} & \cdots &  e^{\im 2\pi (q-1)/q} f_{q-1}  
\end{matrix}
\right]\\
\nu^0&=&
\left[
\begin{matrix}
\mu_{0}  &e^{-\im 2\pi/q}  \mu_{1} &\cdots &e^{-\im 2\pi (p-1)/q}\mu_{p-1}
\end{matrix}
\right]\label{29'}
\eeq
where $f_j, g_j\in \IC^{n\times m/p}$ and $\mu_i, \nu_i\in \IC^{m\times m/p}$. For this sampling scheme, the same relation \eqref{10.1} holds. 
\end{ex}

\commentout{
\begin{rmk}\label{rmk9.2}
While
the ambiguity in Example \ref{ex4} persists under 
the prior information $f_0$ (since $g_0=f_0$), the construction $g(\bn)=f(\bn) \exp(\im \bw\cdot\bn)$
may not be admissible under some partial prior information on $f$ unless the prior is of the form
that $f=0$ at certain locations (in that case $g=0$ at the same locations and the ambiguity persists). 
\end{rmk}
}

A 2D version is as follows. 
\medskip
\begin{ex}\label{ex5} For $q=3, \tau=m/2$, let
 \beq
 f&=&
\lt[\begin{matrix}
f_{00}&f_{10}&f_{20}\\
f_{01}& f_{11}  & f_{21}   \\
f_{02} & f_{12} & f_{22}
\end{matrix}\rt]\\
 g&=&
\left[
\begin{matrix}
f_{00}  & e^{\im 2\pi/3} f_{10} & e^{\im 4\pi/3} f_{20}   \\
e^{\im 2\pi/3} f_{01} &e^{\im 4\pi/3} f_{11} &f_{21}  \\
e^{\im 4\pi/3} f_{02}  & f_{12}& e^{\im 2\pi/3}  f_{22}  
\end{matrix}
\right]
\eeq
be  the object and its reconstruction, respectively, where $f_{ij}, g_{ij}\in \IC^{n/3\times n/3}$.  Let 
 \beq\label{vio}
\mu^{kl}=
\left[
\begin{matrix}
\mu^{kl}_{00}  & \mu^{kl}_{10}    \\
\mu^{kl}_{01}&\mu^{kl}_{11}
\end{matrix}
\right],\quad \nu^{kl}=
\left[
\begin{matrix}
\mu^{kl}_{00}  &e^{-\im 2\pi/3}  \mu^{kl}_{10}    \\
e^{-\im 2\pi/3} \mu^{kl}_{01}&e^{-\im 4\pi/3} \mu^{kl}_{11}
\end{matrix}
\right], 
\eeq
$ k,l=0,1,2,$
be the probe and reconstruction, respectively, where $\mu^{kl}_{ij}, \nu^{kl}_{ij}\in \IC^{n/3\times n/3}$. 

It is verified easily that $\nu^{ij}\odot g^{ij}=e^{\im (i+j)2\pi/3} \mu^{ij}\odot f^{ij}.$
\end{ex}

}

\section{Raster scan ambiguities}
\label{sec:recovery}
In this section we give a complete characterization of the raster scan ambiguities other than the scaling factor
and the affine phase ambiguity \eqref{lp1}-\eqref{lp2}, including the arithmetically progressing phase factor  inherited from the block phases and
the raster grid pathology which
has  a $\tau$-periodic structure of $\tau\times\tau$ degrees of freedom.  We will
use the notation in Section \ref{sec:lattice}.

Before we state the general result. Let us consider two simple examples to illustrate each
type of ambiguity separately.

The first example shows an ambiguity resulting from the arithmetically progressing block phases 
which make positive
and negative imprints on the object and phase estimates, respectively. 

\medskip
\begin{ex}\label{ex6} For $q=3, \tau=m/2$, let
 \beqn
 f&=&
\lt[\begin{matrix}
f_{00}&f_{10}&f_{20}\\
f_{01}& f_{11}  & f_{21}   \\
f_{02} & f_{12} & f_{22}
\end{matrix}\rt]\\
 g&=&
\left[
\begin{matrix}
f_{00}  & e^{\im 2\pi/3} f_{10} & e^{\im 4\pi/3} f_{20}   \\
e^{\im 2\pi/3} f_{01} &e^{\im 4\pi/3} f_{11} &f_{21}  \\
e^{\im 4\pi/3} f_{02}  & f_{12}& e^{\im 2\pi/3}  f_{22}  
\end{matrix}
\right]
\eeqn
be  the object and its reconstruction, respectively, where $f_{ij}\in \IC^{n/3\times n/3}$.  Let 
 \beqn
\mu^{kl}=
\left[
\begin{matrix}
\mu^{kl}_{00}  & \mu^{kl}_{10}    \\
\mu^{kl}_{01}&\mu^{kl}_{11}
\end{matrix}
\right],\quad \nu^{kl}=
\left[
\begin{matrix}
\mu^{kl}_{00}  &e^{-\im 2\pi/3}  \mu^{kl}_{10}    \\
e^{-\im 2\pi/3} \mu^{kl}_{01}&e^{-\im 4\pi/3} \mu^{kl}_{11}
\end{matrix}
\right], 
\eeqn
$ k,l=0,1,2,$
be the $(k,l)$-th shift of the probe and estimate, respectively, where $\mu^{kl}_{ij}\in \IC^{n/3\times n/3}$. 

Let $f^{ij}$ and $g^{ij}$ be the part of the object and estimate illuminated by $\mu^{ij}$ and $\nu^{ij}$, respectively. It is verified easily that $\nu^{ij}\odot g^{ij}=e^{\im (i+j)2\pi/3} \mu^{ij}\odot f^{ij}.$
\end{ex}

The next example illustrates the periodic artifact called raster grid pathology. 

\medskip
\begin{ex}\label{ex5} For $q=3, \tau=m/2$ and any  $\psi\in \IC^{{n\over 3}\times {n\over 3}}$, let 
 \beqn
 f&=&
\lt[\begin{matrix}
f_{00}&f_{10}&f_{20}\\
f_{01}& f_{11}  & f_{21}   \\
f_{02} & f_{12} & f_{22}
\end{matrix}\rt]\\
 g&=&
\left[
\begin{matrix}
e^{-\im \psi}\odot f_{00}  & e^{-\im \psi} \odot f_{10} & e^{-\im \psi} \odot f_{20}   \\
e^{-\im \psi}\odot f_{01} &e^{-\im\psi} \odot f_{11} & e^{-\im\psi}\odot f_{21}  \\
e^{-\im \psi}\odot  f_{02}  & e^{-\im\psi}\odot  f_{12}& e^{-\im \psi} \odot  f_{22}  
\end{matrix}
\right]
\eeqn
be  the object and its reconstruction, respectively, where $f_{ij}\in \IC^{n/3\times n/3}$.  Let 
 \beqn
\mu^{kl}=
\left[
\begin{matrix}
\mu^{kl}_{00}  & \mu^{kl}_{10}    \\
\mu^{kl}_{01}&\mu^{kl}_{11}
\end{matrix}
\right],\quad \nu^{kl}=
\left[
\begin{matrix}
e^{\im\psi}\odot \mu^{kl}_{00}  &e^{\im \psi} \odot \mu^{kl}_{10}    \\
e^{\im \psi}\odot \mu^{kl}_{01}&e^{\im \psi} \odot \mu^{kl}_{11}
\end{matrix}
\right], 
\eeqn
$ k,l=0,1,2,$
be the $(k,l)$-th shift of the probe and estimate, respectively, where $\mu^{kl}_{ij}\in \IC^{n/3\times n/3}$. 

Let $f^{ij}$ and $g^{ij}$ be the part of the object and estimate illuminated by $\mu^{ij}$ and $\nu^{ij}$, respectively. It is verified easily that $\nu^{ij}\odot g^{ij}= \mu^{ij}\odot f^{ij}.$
\end{ex}

The combination of the above two types of ambiguity gives rise to the general
ambiguities for blind ptychography with the raster scan as stated next. 

\medskip
\begin{thm} \label{thm:recovery}  Suppose that $\supp(f)=\IZ_n^2$. 
Consider the raster scan $\cT$ and suppose that an object estimate $g$ and a probe estimate $\nu^{00}$ satisfy the relation
 \beq
 \label{same}
\nu^{kl}\odot g^{kl}=e^{\im\theta_{kl}}\mu^{kl}\odot f^{kl},\quad\theta_{kl}=\theta_{00}+\br\cdot(k,l)
\eeq
as given by Theorem \ref{lem2}. 

Then the following statements hold. 

{\bf (I).} For $\tau\le m/2$, if 
\beq
\label{under1}
\nu^{00}_{00}=e^{\im\psi}\odot\mu^{00}_{00},\quad \psi\in \IC^{\tau\times\tau}, 
\eeq
then  
\beq
\label{a1}
\nu^{00}_{kl}&=&e^{-\im\br\cdot (k,l)}e^{\im\psi}\odot\mu^{00}_{kl},\quad k,l=0,\dots,p-1\\
 g_{kl}&=& e^{\im \theta_{00}}e^{\im\br\cdot (k,l)}e^{-\im\psi}\odot f_{kl},\quad k,l=0,\dots,q-1.\label{a2}
\eeq

{\bf (II).}  For $\tau>m/2$, 
if \beq
\label{over2}
\left[
\begin{matrix}
\nu^{00}_{00}  & \nu^{00}_{10}  \\
\nu^{00}_{01}&\nu^{00}_{11}
\end{matrix}\right] =e^{\im\psi}\odot
\left[
\begin{matrix}
\mu^{00}_{00}  & \mu^{00}_{10}  \\
\mu^{00}_{01}&\mu^{00}_{11}\\
\end{matrix}
\right]
\eeq
for some  
\[
\psi=
\lt[\begin{matrix}
\psi_{00}& \psi_{10}\\
\psi_{01}&\psi_{11}
\end{matrix}\rt]
\in \IC^{\tau\times\tau},
\]
then 
\beq\label{84}
\left[
\begin{matrix}
g^{kl}_{00}  & g^{kl}_{10}   \\
g^{kl}_{01}&g^{kl}_{11}
\end{matrix}
\right]
=e^{\im \theta_{00}}e^{\im\br\cdot (k,l)}e^{-\im\psi}\odot
\left[
\begin{matrix}
f^{kl}_{00}  & f^{kl}_{10}   \\
f^{kl}_{01}&f^{kl}_{11}
\end{matrix}
\right]
\eeq
for all $k,l=0,\dots,q-1$. 
Moreover, 
\beq
\label{a7}
\nu^{00}_{2j}&=&e^{-\im r_1}e^{\im\psi_{0j}}\odot \mu^{00}_{2j}, \quad j=0,1\\
\nu^{00}_{j2}&=&e^{-\im r_2}e^{\im\psi_{j0}}\odot \mu^{00}_{j2}, \quad j=0,1\label{a8}\\
\nu^{00}_{22}&=& e^{-\im (r_1+r_2)}e^{\im\psi_{00}}\odot \mu^{00}_{22}\label{a9}
\eeq
and hence
\beq
\label{a10}
g^{kl}_{2j}&=&e^{\im\theta_{00}}e^{\im \br\cdot ({k+1,l})}e^{-\im\psi_{0j}}\odot f^{kl}_{2j}, \quad j=0,1\\
g^{kl}_{j2}&=&e^{\im\theta_{00}}e^{\im\br\cdot ({k,l+1})}e^{-\im\psi_{j0}}\odot f^{kl}_{j2}, \quad j=0,1\label{a11}\\
g^{kl}_{22}&=& e^{\im\theta_{00}}e^{\im\br\cdot ({k+1,l+1})}e^{-\im\psi_{00}}\odot f^{kl}_{22}. \label{a12}
\eeq

\end{thm}
\medskip
\begin{rmk}
Since $\psi$ is any complex $\tau\times \tau$ matrix, \eqref{under1} and \eqref{over2} represent the maximum degrees of ambiguity
over the respective initial sub-blocks. This ambiguity is transmitted to other sub-blocks, forming periodic artifacts called the raster grid pathology.

On top of the periodic artifacts, there is the non-periodic ambiguity inherited from the affine block phase profile. 
 The non-periodic arithmetically progressing  ambiguity  is different from  the affine phase ambiguity \eqref{lp1}-\eqref{lp2} as they manifest on different scales:
the former on the block scale  while the latter on the pixel scale. 
\end{rmk}
\begin{proof} 

{\bf (I).}  For $\tau\le m/2$, recall the decomposition 
\beqn
\nu^{kl}&=&
\left[
\begin{matrix}
\nu^{kl}_{00}  & \nu^{kl}_{10} &\cdots&  \nu^{kl}_{p-1,0}   \\
\nu^{kl}_{01}&\mu^{kl}_{11}&\cdots& \nu^{kl}_{p-1,1}\\
\vdots&\vdots&\vdots&\vdots\\
\nu^{kl}_{0,p-1}  &\nu^{kl}_{1,p-1}&\cdots& \nu^{kl}_{p-1,p-1} 
\end{matrix}
\right],\quad 
 g=
\left[
\begin{array}{ccc}
g_{00}  & \ldots  & g_{q-1,0}   \\
\vdots  & \vdots  &\vdots   \\
g_{0,q-1}  &\ldots   &  g_{q-1, q-1}  
\end{array}
\right], \eeqn
with $\nu^{kl}_{ij},g_{ij}\in \IC^{m/p\times m/p}$, in analogy to \eqref{mu} and  \eqref{fp}.

\[
g_{00}=e^{\im\theta_{00}} e^{-\im\psi}\odot f_{00}
\]
 by restricting  \eqref{same} to $\cM^{00}_{00}$. 
 
 For $\bn\in \cM_{00}^{10}$,  we have 
\beqn
\nu_{00}^{10}\odot g_{10}=e^{\im\theta_{10}} \mu^{10}_{00}\odot f_{10},
\eeqn
by \eqref{same}, and 
\[
\nu_{00}^{10}(\bn)=\nu^{00}_{00}(\bn-(\tau,0))=(e^{\im \psi}\odot\mu_{00}^{00})(\bn-(\tau,0))=(e^{\im\psi}\odot \mu^{10}_{00})(\bn)
\] by \eqref{under1}.
Hence 
\[
g_{10}=e^{\im\theta_{10}}e^{-\im\psi}\odot f_{10}
\]
 implying 
\[
\nu^{00}_{10} \odot g_{10}=e^{\im\theta_{10}}e^{-\im\psi}\nu^{00}_{10}\odot f_{10}=e^{\im\theta_{00}}\mu^{00}_{10}\odot f_{10}
\]
by \eqref{same} and consequently 
\[
\nu_{10}^{00}=e^{\im\theta_{00}}e^{-\im\theta_{10}}e^{\im\psi}\mu^{00}_{10}.
\]

Repeating the same argument for the adjacent blocks in both directions, we obtain
\beqn
\nu^{00}_{kl}&=&e^{\im\theta_{00}}e^{-\im\theta_{kl}}e^{\im\psi}\odot\mu^{00}_{kl}\\
g_{kl}&=&e^{\im\theta_{kl}} e^{-\im\psi}\odot f_{kl}
\eeqn
which are equivalent to \eqref{a1} and \eqref{a2} in view of the block phase profile in \eqref{a3}. 

{\bf (II).} 
First recall 
\beqn
\mu^{kl}&=&
\left[
\begin{matrix}
\mu^{kl}_{00}  & \mu^{kl}_{10} &  \mu^{kl}_{20}   \\
\mu^{kl}_{01}&\mu^{kl}_{11}&\mu^{kl}_{21}\\
\mu^{kl}_{02}  &\mu^{kl}_{12}& \mu^{kl}_{22} 
\end{matrix}
\right],  \quad
g=\bigvee_{k,l=0}^{q-1} g^{kl},\quad    g^{kl}= 
\left[
\begin{matrix}
g^{kl}_{00}  & g^{kl}_{10} &  g^{kl}_{20}   \\
g^{kl}_{01}&g^{kl}_{11}&g^{kl}_{21}\\
g^{kl}_{02}  &g^{kl}_{12}& g^{kl}_{22} 
\end{matrix}
\right]
\eeqn
in analogy to \eqref{mu2} and \eqref{fp2}. 

Since 
\beq
\label{trans}
\nu^{kl}(\bn)=\nu^{00}(\bn-\tau(k,l))),\quad \mu^{kl}(\bn)=\mu^{00}(\bn-\tau (k,l)),
\eeq
  \eqref{84} follows 
from   \eqref{over2} and \eqref{same}. 

By \eqref{over2} and restricting \eqref{same} to $\cM^{10}_{0j}, j=0,1,$
 we obtain
\[
g^{00}_{2j}=g^{10}_{0j}=e^{\im\theta_{10}} e^{-\im\psi_{0j}}\odot f^{10}_{0j}=e^{\im\theta_{10}} e^{-\im\psi_{0j}}\odot f^{00}_{2j},\quad j=0,1,
\]
which implies by \eqref{same}
\beqn
\nu^{00}_{2j}&=&e^{\im\theta_{00}}e^{-\im\theta_{10}}e^{\im\psi_{0j}}\odot \mu^{00}_{2j}, \quad j=0,1,\\
\nu^{00}_{j2}&=&e^{\im\theta_{00}}e^{-\im\theta_{01}}e^{\im\psi_{j0}}\odot \mu^{00}_{j2}, \quad j=0,1,
\eeqn
and consequently  \eqref{a7} and \eqref{a8}. 

By \eqref{trans} and restricting  \eqref{same} to $\cM^{kl}_{2j}, \cM^{kl}_{j2}, j=0,1,$ 
we have  \eqref{a10} and \eqref{a11}. 

For \eqref{a12} with $(k,l)=(0,0)$, the block $\cM^{10}_{02}=\cM^{00}_{22}$ is masked by $\mu^{10}_{02}$, a translate  of $\mu^{00}_{02}$. By restricting \eqref{same} to $\cM^{10}_{02}$, 
\beq
\label{r22'}
g^{00}_{22}=g^{10}_{02}=e^{\im (\theta_{10}+\theta_{01}-\theta_{00})}e^{-\im \psi_{00}}\odot f^{00}_{22}.
\eeq
which is equivalent to \eqref{a12} with $(k,l)=(0,0)$. Then  \eqref{same} and \eqref{r22'} imply 
\beq
\label{r22}
\nu^{00}_{22}=e^{\im (\theta_{00}-\theta_{10})} e^{\im (\theta_{00}-\theta_{01})} e^{\im\psi_{00}}\odot \mu^{00}_{22}
\eeq
which is equivalent to \eqref{a9}. 

For \eqref{a12} with general $k,l$, by restricting  \eqref{same} to $\cM^{kl}_{22}$  and 
 \eqref{r22} we have 
\[
g^{kl}_{22}=e^{\im\theta_{kl}}e^{\im (\theta_{10}-\theta_{00})} e^{\im (\theta_{01}-\theta_{00})} e^{-\im\psi_{00}}\odot f^{kl}_{00}
\]
and hence  \eqref{a12}. 
 \end{proof}
 
When $\tau=1$, the non-periodic, arithmetically progressing ambiguity and the affine phase ambiguity become the same. In addition, for $\tau=1$ the raster grid pathology becomes a constant phase factor
which can be ignored \cite{Iwen}. 

\medskip

\begin{cor}\label{cor:raster}
If $\tau=1$ (i.e. $q=n, p=m$) and \eqref{same} holds, then 
the probe and the object can be uniquely and simultaneously determined. 
\end{cor}
\begin{proof}
For $\tau=1$, $\mu_{00}$ consists of just one pixel and $\psi$ is a number. Hence
$\mu^{00}=\nu^{00}$ up to a constant phase factor and \eqref{a1}-\eqref{a2} then imply that
the affine phase ambiguity is the only ambiguity modulo the constant phase factor. 

\end{proof}

\section{Slightly perturbed raster scan}\label{sec:mix}

\begin{figure}
\centering
\subfigure[Perturbed grid  \eqref{perturb}]{\includegraphics[width=6cm,height=5.7cm]{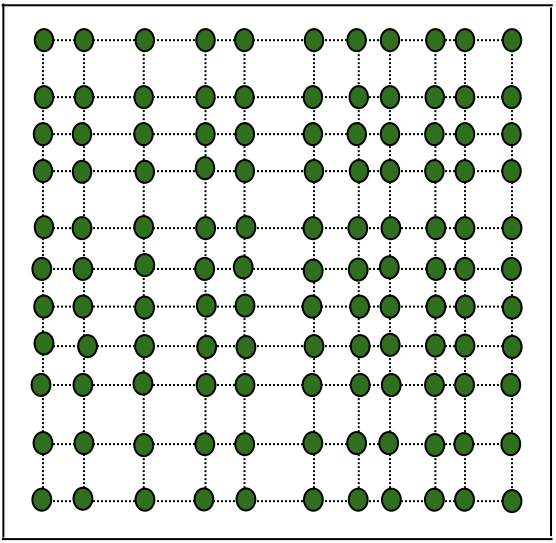}}\hspace{1cm}
\subfigure[Perturbed grid \eqref{perturb2}]{\includegraphics[width=6cm,height=5.7cm]{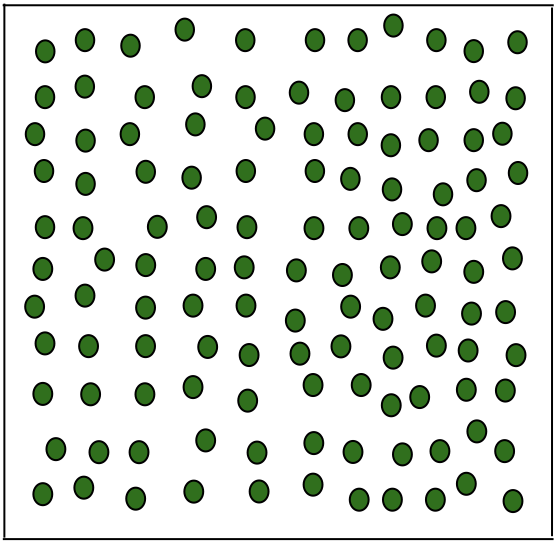}}
\label{fig:irregular}
\end{figure}

In this section, we demonstrate a simple way for removing
all the raster scan ambiguities except for the scaling factor and
the affine phase ambiguity.

For the rest of the paper, we assume that $f$ does not vanish in $\IZ_n^2$. 

We consider the perturbed raster scan (Fig. \ref{fig:irregular}(a))
\beq
\label{perturb}
\bt_{kl}=\tau (k,l)+  (\delta^1_{k},\delta^2_{l}),\quad k,l=0,\dots, q-1
\eeq
where $ \delta^1_{k}, \delta^2_{l}$ are small integers relative to  $\tau$ and $m-\tau$ (see Theorem \ref{thm:mix} for details). 
More general than \eqref{perturb} is the perturbed grid pattern (Fig. \ref{fig:irregular}(b)):
\beq
\label{perturb2}
\bt_{kl}=\tau (k,l)+  (\delta^1_{kl},\delta^2_{kl}),\quad k,l=0,\dots, q-1, 
\eeq
which is harder to analyze and implement in experimental practice (we will only present numerical simulation for it). 
Without loss of generality we set $\delta^1_0=\delta_0^2=0$ and hence $\bt_{00}=(0,0)$. 

Let us express the probe and object errors in terms of  
\beq
\label{100}
\nu^{00}(\bn)/\mu^{00}(\bn)&:=&\alpha(\bn) \exp{(\im \phi(\bn))},\quad\bn\in\cM^{00}\\
h(\bn)&:=& \ln g(\bn)-\ln f(\bn),\quad \bn\in \IZ_n^2,
\eeq
where we assume $\alpha(\bn)\neq 0$ for all $\bn\in \cM^{00}$,
 and rewrite \eqref{7.1} as  
\beq
\label{200.1}
h(\bn+  \bt) &=&\im \theta_\bt-\ln\alpha(\bn)-\im \phi(\bn) \mod\im 2\pi,
\eeq
for $ \bn \in \cM^{00}$.

By  \eqref{200.1} with $\bt=(0,0)$,
\beq
\label{200.4'}
h(\bn)=\im \theta_{00}-\ln\alpha(\bn)-\im \phi(\bn), \quad \forall \bn \in \cM^{00}
\eeq
and hence for all $ \bt\in \cT$ and $\bn\in \cM^{00}$
\beq
\label{200.4}
h(\bn+  \bt)-h(\bn)=\im\theta_\bt-\im\theta_{00}\mod\im 2\pi. 
\eeq

We wish to generalize such a relationship to the case where $\bt$ in \eqref{200.44} is replaced
by $\be_1=(1,0)$ and $\be_2=(0,1)$.



\subsection{A simple perturbation}
Let us first study the simple example of the two-shift perturbation to the  raster-scan with
 $\delta^1_{2}=\delta^2_{2}=-1$ but all other $\delta^j_k=0$, i.e. $\bt_{kl}=\tau(k,l)$ for $(k,l)\neq (2,0), (0,2)$. Then 
\beq
 &&h(\bn+  2\bt_{10}-\bt_{20})=h(\bn+  (1,0))\label{17.1}\\
 &&h(\bn+  2\bt_{01}-\bt_{02})=h(\bn+  (0,1)). \label{18.1}
\eeq
There are several routes of reduction from $(1,0)$ to $(0,0)$ via the shifts in $\cT$. For example, 
we can proceed from $(1,0)=2\bt_{10}-\bt_{20}$ to $(0,0)$
along the path 
\beq
\label{reduce2}
(2\bt_{10}-\bt_{20})\longrightarrow (\bt_{10}-\bt_{20}) \longrightarrow \bt_{10}\longrightarrow (0,0)
\eeq
by repeatedly applying \eqref{200.4}
where the direction of the second step is to be reversed since $-\bt_{20}\not\in\cT$ ($\cT$ is no longer a group even under the periodic boundary condition). The direction is important for keeping  track of the domain
of validity of \eqref{200.4} along the path.
Hence for all 
\beq
\label{43}\bn\in (\cM^{00}+\bt_{20}- \bt_{10})\cap \cM^{00} 
\eeq
we have \beqn
h(\bn+  2\bt_{10}-\bt_{20})&=&h(\bn+  \bt_{10}-\bt_{20})+  \im\theta_{10}-\im\theta_{00}
\\
&=&h(\bn+\bt_{10})+\im\theta_{10}-\im\theta_{20}\\
&=&h(\bn)+2\im\theta_{10}-\im\theta_{20}-\im\theta_{00}
\eeqn
\commentout{
\beq
\label{reduce2}
(2\bt_{10}-\bt_{20})\longrightarrow (\bt_{10}-\bt_{20}) \longrightarrow -\bt_{20}\longrightarrow (0,0)
\eeq
by repeatedly applying \eqref{200.4}
where the direction of the second step is to be reversed since $-\bt_{20}\not\in\cT$ ($\cT$ is no longer a group even under the periodic boundary condition). 
We  must also keep track of the domain
of validity for \eqref{200.4} along the path.
Hence for all 
\beqn
\bn\in (\cM^{00}+\bt_{20}- \bt_{10})\cap (\cM^{00}+ \bt_{20}) 
\eeqn
we have \beqn
h(\bn+  2\bt_{10}-\bt_{20})&=&h(\bn+  \bt_{10}-\bt_{20})+  \im\theta_{10}-\im\theta_{00}
\\
&=&h(\bn-\bt_{20})+2\im\theta_{10}-2\im\theta_{00}
\\
&=&h(\bn)+2\im\theta_{10}-\im\theta_{20}-\im\theta_{00}
\eeqn
}
 and hence
\beq
\label{200.10} \label{57'}
h(\bn+  (1,0))=h(\bn)+\im \Delta^1,\quad \Delta^1:= 2\theta_{10}-\theta_{20}-\theta_{00}
\eeq
modulo $\im2\pi$. 
\commentout{ where 
\beqn
D&:=& \lt[\cM^{00}\cap (\cM^{00}-(1,0))\cap  (\cM^{00}-\bt_{10})\rt]\\
&&\bigcup \lt[(\cM^{00}-(1,0)+  \bt_{10})\cap (\cM^{00}+  \bt_{20})\rt] \\
&=& \lt[\cM^{00}\cap(\cM^{00}-(1,0))\cap  (\cM^{00}-(\tau,0))\rt]\\
&&\bigcup  \lt[(\cM^{00}+  (\tau-1,0))\cap (\cM^{00}+  (2\tau-1,0))\rt].\eeqn
Note that we have used the identity $-(1,0)+  2\bt_{10}=\bt_{20}$. 
}

Let us consider another alternative route for reduction:
 \beq\label{reduce1}
(2\bt_{10}-\bt_{20}) \longrightarrow 2\bt_{10}\longrightarrow \bt_{10} \longrightarrow (0,0) 
\eeq
where the proper direction for the first step in applying \eqref{200.4} is reversed. Keeping  track of the domain
of validity along the path, we have 
\beqn
h(\bn+  2\bt_{10}-\bt_{20})&=&h(\bn+  2\bt_{10})-\im\theta_{20}+  \im\theta_{00}\\
&=&h(\bn+  \bt_{10})+\im\theta_{10}-\im\theta_{20}\\
&=&h(\bn)+\im \Delta^1
\eeqn
for all  
\beq
\label{45}
 \bn\in (\cM^{00}-2\bt_{10}+\bt_{20})\cap  (\cM^{00}-\bt_{10})\cap \cM^{00}.
 \eeq
 In summary, \eqref{200.10} holds for all $\bn$ in the union of \eqref{43} and \eqref{45}, i.e.
 \beqn
D^1
&=&(\lb 0, m-\tau-1\rb \cup\lb \tau-1,m-1\rb)\times \lb 0,m-1\rb.
 \eeqn

Clearly including other routes for reducing $2\bt_{10}-\bt_{20}$ to
$(1,0)$ in $D^1$ can enlarge the domain of validity for \eqref{200.10}.
For simplicity of argument, we omit them here.

By repeatedly applying \eqref{200.4} we have the following result.
\begin{prop}\label{prop:key}
The relation  \eqref{200.10}
\commentout{
 \beq
 \label{57}
h(\bn+(1,0))
&=&h(\bn)+\im \Delta^1\mod \im2\pi
\eeq
}
holds true in the set
\beq
\label{59}
\bigcup_{\bt\in \cT}\lt[\bt+D^1\cap \cM^{00}\cap (\cM^{00}-\be_1)\rt] 
\eeq
which contains $\IZ_n^2$ if
\beq
 \label{A1}
 \tau\le (m-2)\wedge [(m+1)/2]. 
 \eeq

\end{prop}
\begin{proof}
For  $\bn\in D^1\cap \cM^{00}\cap (\cM^{00}-\be_1)$, we have 
\beq
h(\bn+\bt)&=&h(\bn+\be_1)-\im \Delta^1
+\im\theta_\bt-\im\theta_{00},\label{58'}
\eeq
by \eqref{57'} and \eqref{200.4}. 

Hence,  by \eqref{200.4} and \eqref{58'}, 
\beqn
h(\bn+\be_1+\bt)&=&h(\bn+\be_1)+\im\theta_\bt-\im\theta_{00}\\
&=& h(\bn+\bt)+\im \Delta^1.
\eeqn
In other words, \eqref{57'} has been extended to $\bt+D^1\cap \cM^{00}\cap (\cM^{00}-\be_1)$. 
Taking the union over all shifts, we obtain \eqref{59}.

For the second part of the proposition, let us write the set \eqref{59} explicitly as
\[
\bigcup_{k,l=0}^{q-1}\lt\{\tau(k,l)+ \lt[(\lb 0, m-\tau-1\rb \cup\lb \tau-1,m-1\rb)\cap \lb 0, m-2\rb\rt]\times \lb 0,m-1\rb\rt\}.
\]
Note that 
\beqn
 \lt[(\lb 0, m-\tau-1\rb \cup\lb \tau-1,m-1\rb)\cap \lb 0, m-2\rb\rt] &=&\lb 0, m-\tau-1\rb\cup \lb \tau-1, m-2\rb\\
 &=& \lb 0, m-2\rb
 \eeqn
under  $m-\tau-1\ge \tau-2$ or,  equivalently,  \eqref{A1}. To complete the argument, observe that 
the adjacent rectangles among
\[
(\tau(k, l)+\lb 0, m-2\rb )\times \lb 0, m-1\rb,\quad k,l=0,\dots, q-1,
\]
have zero gap if $\tau\le m-2$.

\end{proof}

\commentout{
 We now extend  \eqref{200.10} from $D^1$  to all $\bn\in \IZ_n^2$ by applying \eqref{200.4} for all $\bt\in\cT$. In other words,  we wish to establish
\beq
\label{cover}
\IZ_n^2= \bigcup_{\bt\in \cT} (D^1+\bt)
\eeq
where the shift is under the periodic boundary condition on $\IZ_n^2$. 
\commentout{
Let $\lb k, l\rb$ denote the integers between, and including $k\le l\in \IZ$.  First observe that 
\[
D= \lt(\lb 0, m-\tau-1\rb\cup \lb 2\tau-1,m+\tau-2\rb\rt)\times \lb 0, m-1\rb 
\]
and hence $D\neq\emptyset $ if $\tau<m$.
}

To this end, it suffices to show that, for any $R$ congruent to $\cM^{00}=\IZ_m^2,$
\[
R\subseteq  \bigcup_{\bt\in \cT} (D^1+\bt)
\]
since $\IZ_n^2$ can then be covered by $\{R+\bt: \bt\in \cT\}$.

To this end,  we note that  the simple condition, $\tau\le (m+1)/2,$
implies $m-\tau-1\ge\tau-1-1$ and hence $D^1=\lb 0, m-1\rb^2=\IZ_m^2$.  The condition $\tau\le (m+1)/2$ means that the overlap ratio between the adjacent probes is at least $50\%$. 
}

By the same argument under \eqref{A1},  it follows  from \eqref{18.1} that for all $\bn\in \IZ_n^2 $
\beq
\label{prop2}
h(\bn+  (0,1))=h(\bn)+\im \Delta^2 \mod \im 2\pi,\quad \Delta^2:=2\im\theta_{01}-\im\theta_{02}-\im\theta_{00}.
\eeq
In conclusion, 
\beq\label{101}
h(\bn)=h(0)+\im \bn\cdot \br\mod \im 2\pi, \quad \forall\bn\in \IZ_n^2,
\eeq
where $\br=(\Delta^1,\Delta^2)$.

\commentout{
Clearly,  \eqref{200.1} and \eqref{101} imply that  for all $\bn\in \cM^{00}$ 
\beq
\label{200.11}
\alpha(\bn) &=&\alpha(0), \\
\phi(\bn)&=&\phi(0)-\bn\cdot \br\mod 2\pi.\label{phase-error}
\eeq

Having shown that the probe phase error $\phi$ has an affine profile \eqref{phase-error}, let us now turn to the
block phases $\theta_k$. 

By \eqref{200.1}, 
\beq
\label{200.113}
h(\bn+  \bt_{kl})
&=&\im \theta_{kl}-\ln\alpha(0)-\im \phi(\bn)  \mod\im 2\pi,\quad \bn\in \cM^{00},
\eeq
for $k, l=0,\cdots, q-1$.

Substituting $\bn=(0,0)$ into \eqref{200.113} we have
\[
h(\bt_{kl})
=\im \theta_{kl}-\ln\alpha(0)-\im \phi(0)  \mod \im 2\pi.
\]
On the other hand, with the replacement $\bt_{kl}\to \bt_{00}$ and $\bn\to\bt_{kl}$ in \eqref{200.113}, we have
\[
h(\bt_{kl})=\im\theta_{00}-\ln\alpha(0)-\im \phi(\bt_{kl})  \mod\im 2\pi
\]
and hence 
\beq
\label{200.119}
\theta_{kl}-\theta_{00}=\phi(0)-\phi(\bt_{kl}),
\eeq
for $k,l=0,\cdots, q-1$.

It follows from \eqref{200.119} and  \eqref{phase-error} that 
\beq
\label{200.117}
{\theta_{kl}-\theta_{00}}
+  \bt_{kl}\cdot\br=0\mod 2\pi,
\eeq
for $k,l=0,\cdots,q-1.$

To summarize, we have shown that the scaling factor in \eqref{200.11} and the affine phase ambiguity, in \eqref{101} and \eqref{phase-error}, are
the only ambiguities in the slightly perturbed raster scan. 
}

\subsection{General perturbation}

Next we consider more general perturbations $\{\delta^i_k\}$ to the raster scan
and derive \eqref{101}.

Let us rewrite \eqref{200.4} in a different form:
Subtracting the respective \eqref{200.4} for $\bt$ and $\bt'$, we obtain
the equivalent form
\beq
\label{200.44'}
h(\bn+  \bt)-h(\bn+\bt')=\im\theta_{\bt}-\im\theta_{\bt'}\mod\im 2\pi,
\eeq
for any $\bn\in \cM^{00}$ and $\bt, \bt'\in \cT$, 
which  can also be written as
\beq\label{200.44}
h(\bn+  \bt-\bt')=h(\bn)+\im(\theta_{\bt}-\theta_{\bt'})\mod\im 2\pi,
\eeq
for $\bn\in \cM^{\bt'}$ by shifting the argument of $h$.

Consider  the triplets of shifts
\[
(\bt_{kl},\bt_{k+1,l},\bt_{k+2,l}),\quad (\bt_{kl},\bt_{k,l+1},\bt_{k,l+2})
\]
for which we have
\beq
\nn{2(\bt_{k+1,l}-\bt_{kl})-(\bt_{k+2,l}-\bt_{kl})}
&=&(2\delta_{k+1}^1-\delta^1_{k}-\delta^1_{k+2},0):=(a^1_{k},0),\label{17.1'}\\
\nn{2(\bt_{k,l+1}-\bt_{kl})-(\bt_{k,l+2}-\bt_{kl})}
&=& (0, 2\delta_{l+1}^2-\delta^2_{l}-\delta^2_{l+2}):=(0,a^2_{l}).\label{18.1'}
\eeq

Analogous to \eqref{reduce1} and \eqref{reduce2} the paths of reduction 
\beqn
{(2\bt_{k+1,l}-\bt_{kl}-\bt_{k+2,l})\longrightarrow 2(\bt_{k+1,l}-\bt_{kl})}
\longrightarrow (\bt_{k+1,l}-\bt_{kl}) \longrightarrow (0,0)
\eeqn
and
\beqn
{(2\bt_{k+1,l}-\bt_{kl}-\bt_{k+2,l})\longrightarrow (\bt_{k+1,l}-\bt_{k+2,l})}\longrightarrow (\bt_{k+1,l}-\bt_{kl})\longrightarrow (0,0)
\eeqn
lead to  \beq
\label{200.100}
{h(\bn+(a^1_{k},0))}
&=& h(\bn)+2\im\theta_{k+1,l}-\im\theta_{k+2,l}-\im\theta_{kl}\mod \im 2\pi
\eeq
for all $\bn\in D^1_{kl}$ where 
\beqn
D^1_{kl}&:=&\lt\{\cM^{kl}\cap [\cM^{kl}-2\bt_{k+1,l}+\bt_{k+2,l}+\bt_{kl}]\rt.\lt.\cap [\cM^{kl}-\bt_{k+1,l}+\bt_{kl}]\rt\}\\
&&\bigcup \lt\{\cM^{kl}\cap [\cM^{kl}+\bt_{k+2,l}-\bt_{k+1,l}]\rt\}\\
&=&\lt\{\cM^{kl}\cap [\cM^{kl}-(a^1_{k},0)]\rt.\lt. \cap  [\cM^{kl}-(\tau+\delta^1_{k+1}-\delta^1_k,0)]\rt\}\nn\\
\nn&&\bigcup \lt\{\cM^{kl}\cap [\cM^{kl}+(\tau+\delta^1_{k+2}-\delta^1_{k+1},0)]\rt\}
\eeqn

Likewise, repeatedly applying \eqref{200.44} along the paths,
\beqn
{(2\bt_{k,l+1}-\bt_{kl}-\bt_{k,l+2})\longrightarrow 2(\bt_{k,l+1}-\bt_{kl})}\longrightarrow (\bt_{k,l+1}-\bt_{kl}) \longrightarrow (0,0)
\eeqn
and
\beqn
{(2\bt_{k,l+1}-\bt_{kl}-\bt_{k+2,l})\longrightarrow (\bt_{k+1,l}-\bt_{k+2,l})}\longrightarrow (\bt_{k+1,l}-\bt_{kl})\longrightarrow (0,0)
\eeqn
we get
\beq
\label{200.200}{h(\bn+(0,a^2_{l}))}
&=& h(\bn)+2\im\theta_{k,l+1}-\im\theta_{k,l+2}-\im\theta_{kl} \mod \im 2\pi
\eeq
for $\bn\in D^2_{kl}$ where 
\beqn
D^2_{kl}&:=&\lt\{\cM^{kl}\cap [\cM^{kl}-2\bt_{k,l+1}+\bt_{k,l+2}+\bt_{kl}]\rt.\lt.\cap [\cM^{kl}-\bt_{k,l+1}+\bt_{kl}]\rt\}\\
&&\bigcup \lt\{\cM^{kl}\cap [\cM^{kl}+\bt_{k,l+2}-\bt_{k,l+1}]\rt\}\\
&=&\lt\{\cM^{kl}\cap [\cM^{kl}-(0, a^2_{l})]\rt.\lt. \cap  [\cM^{kl}-(0, \tau+\delta^2_{l+1}-\delta^2_l)]\rt\}\nn\\
\nn&&\bigcup \lt\{\cM^{kl}\cap [\cM^{kl}+(0, \tau+\delta^2_{l+2}-\delta^2_{l+1})]\rt\}.
\eeqn

\commentout{
\beqn
2\bt_{k,l+1}-\bt_{kl}-\bt_{k,l+2}\longrightarrow 2(\bt_{k,l+1}-\bt_{kl})\longrightarrow (0,0)
\eeqn
and
\beqn
2\bt_{k,l+1}-\bt_{kl}-\bt_{k,l+2}\longrightarrow -(\bt_{k,l+2}-\bt_{kl})\longrightarrow (0,0),
\eeqn
we get
\beq
\label{200.200}\lefteqn{h(\bn+(0,a^2_{l}))}\\
&=& h(\bn)+2\im\theta_{k,l+1}-\im\theta_{k,l+2}-\im\theta_{kl} \mod \im 2\pi\nn
\eeq
 for all $\bn\in \IZ_n^2$ provided that 
\beqn
D^2_{kl}&:=&\lt\{\cM^{kl}\cap [\cM^{kl}-(0,a^2_{l})]\rt.\\
&&\nn \lt. \cap  [\cM^{kl}-(0, \tau+\delta^2_{l+1}-\delta^2_l)]\rt\}\\
\nn&&\bigcup\lt\{[\cM^{kl}-(0,a^2_{l})+(0,\tau+\delta^2_{l+1}-\delta^2_l)]\rt.\\
&&\nn \lt. \cap [\cM^{kl}-(0,a^2_l)+2(0,\tau+\delta^2_{l+1}-\delta^2_{l})]\rt\}\neq\emptyset.
\eeqn
}

\begin{lemma}\label{lem:key}
Let $k,l$ be fixed. The relations  \eqref{200.100} and \eqref{200.200} 
hold true in the sets
\beq
\label{400}
\bigcup_{\bt\in \cT}\lt[\bt+D_{kl}^1\cap \cM^{kl}\cap (\cM^{kl}-(a^1_k,0))\rt] 
\eeq
and 
\beq\label{401}
\bigcup_{\bt\in \cT}\lt[\bt+D_{kl}^2\cap \cM^{kl}\cap (\cM^{kl}-(0,a^2_l))\rt], 
\eeq
respectively. Both sets
 contain $\IZ_n^2$ if the following conditions hold:
 \beq
\label{small1}
&&\max_{i=1,2}\{|a^i_k|+\delta_{k+1}^i-\delta_k^i\}\le\tau\\
\label{cover2}
&&2\tau \le m-\max_{i=1,2}\{\delta^i_{k+2}-\delta^i_k\}\\
\label{small2}
&&\max_{k'}\max_{i=1,2}\{|a^i_{k}| +\delta^i_{k'+1}-\delta^i_{k'}\}\le m-1-\tau. 
\eeq
\end{lemma}
\medskip
\begin{rmk}
Proposition \ref{prop:key} corresponds to $(k,l)=(0,0)$ with \eqref{small1}, \eqref{cover2} and \eqref{small2} reduced to
\[
1\le \tau,\quad 2\tau\le m+1,\quad \tau\le m-2, 
\]
respectively. 
\end{rmk}
\medskip
\begin{rmk}
Inequalities \eqref{small1} and \eqref{small2} are smallness conditions for the perturbations relative to the average step size and the overlap between
the adjacent probes. The most consequential condition \eqref{cover2} suggests an average overlap ratio of at least $50\%$, i.e. under-shifted raster scan. 
\end{rmk}
\begin{proof}
The argument follows the same pattern as that for Proposition \ref{prop:key}. 

For  $\bn\in D_{kl}^1\cap \cM^{00}\cap (\cM^{00}-(a^1_k,0))$, we have 
\beqn
h(\bn+\bt)&=&h(\bn+(a^1_k,0))-\im (2\theta_{k+1,l}-\theta_{k+2,l}-\theta_{kl})
+\im\theta_\bt-\im\theta_{00}
\eeqn
by \eqref{200.100} and \eqref{200.4}. 

Hence,  by \eqref{200.4} and \eqref{58'}, 
\beqn
h(\bn+(a^1_k,0)+\bt)&=&h(\bn+(a^1_k,0))+\im\theta_\bt-\im\theta_{00}\\
&=& h(\bn+\bt)+\im (2\theta_{k+1,l}-\theta_{k+2,l}-\theta_{kl}).
\eeqn
Taking the union over all shifts, we obtain the set in \eqref{400}. The case for \eqref{401} is similar. 

For the second part of the proposition, note that
\beq
\label{404}
\cM^{kl}\cap (\cM^{kl}-(a^1_k,0)) &=& \lb 0, m-1-|a^1_{k}|\rb \times \lb 0,m-1\rb\quad \mbox{if}\quad a^1_{k}\ge 0\\
\quad &\mbox{or}& \lb |a^1_{k}|, m-1\rb \times \lb 0,m-1\rb,\quad\mbox{if}\quad a^1_{kl}<0.\nn
\eeq
\commentout{
Hence 
\beqn
D^1_{kl}&\supseteq&(\lb |a^1_{kl}|, m-1-\tau-\delta^1_{k+1}+\delta^1_k\rb \cup\lb \tau+\delta^1_{k+2}-\delta^1_{k+1},m-1\rb)\\
&&\times \lb 0,m-1\rb + (\tau k+\delta^1_k,\tau l+\delta^2_l)
\eeqn
where we assume
}
In the former case  in \eqref{404}  the set \eqref{400} contains
\beq
\label{402}&&\bigcup_{\bt\in \cT}\lt[\bt+ \bt_{kl}+(\lb 0, m-1-\tau-\delta^1_{k+1}+\delta^1_k\rb \cup\lb \tau+\delta^1_{k+2}-\delta^1_{k+1},m-1\rb)\rt.\\
&&\quad\quad \lt.\times \lb 0,m-1\rb \cap \lt(\lb 0, m-1-|a^1_{k}|\rb \times \lb 0,m-1\rb\rt)\rt]\nn
\eeq
under the condition
\beq
\label{407}
 |a^1_k|+\delta_{k+1}^1-\delta_k^1\le \tau \le  m-1-|a_k^1|-\delta^1_{k+1}+\delta^1_k.
\eeq
The set in \eqref{402}  becomes
 \beq
 \label{403}
&&\bigcup_{\bt\in \cT}\lt[\bt+ \bt_{kl}+ \lt(\lb 0, m-1\rb \cap\lb 0, m-1-|a^1_{k}|\rb\rt) \times \lb 0,m-1\rb\rt]\\
&=&\bigcup_{\bt\in \cT}\lt[\bt+ \bt_{kl}+\lb 0, m-1-|a^1_{k}|\rb \times \lb 0,m-1\rb\rt]\nn
\eeq
 under the condition
  \beq
\label{57}
m-1-\tau-\delta^1_{k+1}+\delta^1_k &\ge & \tau+\delta^1_{k+2}-\delta^1_{k+1}-1.
\eeq
The set in \eqref{403} contains $\IZ_n^2$ if for each $ l'$ the adjacent sets among
\[
\tau(k'+k,l'+l)+ (\delta^1_{k'}+\delta^1_k,\delta^2_{l'}+\delta^2_l)+ \lb 0, m-1-|a^1_{k}|\rb \times \lb 0,m-1\rb,
\]
for $k'=0,\dots, q-1,$ have no gap between them, which is the case if
\beq
\label{405}
\tau+\delta^1_{k'+1}-\delta^1_{k'}\le m-1-|a^1_{k}|,\quad \forall k'.
\eeq
Note that \eqref{405} subsumes the second inequality in \eqref{407}.

Likewise for the latter case in \eqref{404} the set in \eqref{400} contains
\beq
\label{406}&&\bigcup_{\bt\in \cT}\lt[\bt+ \bt_{kl}+(\lb |a^1_k|, m-1-\tau-\delta^1_{k+1}+\delta^1_k\rb \cup\lb \tau+\delta^1_{k+2}-\delta^1_{k+1},m-1\rb)\rt.\\
&&\quad\quad \lt.\times \lb 0,m-1\rb \cap \lt(\lb |a^1_k|, m-1\rb \times \lb 0,m-1\rb\rt)\rt]\nn
\eeq
under the condition 
\beq\label{408}
|a_l^2| +\delta^2_{l+1}-\delta^2_l\le \tau\le m-1-|a_l^2| -\delta^2_{l+1}+\delta^2_l.
\eeq
The set in \eqref{406} in turn becomes
 \beq
\label{403'}\bigcup_{\bt\in \cT}\lt[\bt+ \bt_{kl}+\lb |a^1_k|, m-1\rb \times \lb 0,m-1\rb\rt]\nn
\eeq
 under the condition
  \beqn
m-1-\tau-\delta^1_{k+1}+\delta^1_k &\ge & \tau+\delta^1_{k+2}-\delta^1_{k+1}-1.
\eeqn
The set in \eqref{403'} contains $\IZ_n^2$ if for each $ l'$ the adjacent sets among
\[
\tau(k'+k,l'+l)+ (\delta^1_{k'}+\delta^1_k,\delta^2_{l'}+\delta^2_l)+ \lb |a^1_k|, m-1\rb \times \lb 0,m-1\rb,
\]
for $k'=0,\dots, q-1,$ have no gap between them, which is the case under the same condition \eqref{405}
which subsumes the second inequality in \eqref{408}. 

The case with \eqref{401} can be proved by the same argument as  above. 
\end{proof}

Since  $\cM^{kl}$ overlaps with $\cM^{k+1,l}$ and $\cM^{k,l+1}$ which in turn overlap with $\cM^{k+2,l}$ and $\cM^{k,l+2}$, respectively
(and so on), 
the quantities 
\beq
\label{D1}
\Delta^1_k:=2\theta_{k+1,l}-\theta_{k+2,l}-\theta_{kl}\\
\Delta^2_l:=2\theta_{k,l+1}-\theta_{k,l+2}-\theta_{kl}\label{D2}
\eeq
on the righthand side of  \eqref{200.100} and \eqref{200.200} depend only on one index
and we can write


Suppose further that there exist $c^1_k,c^2_l\in \IZ$ such that
\beq
\sum_{k=0}^{q-1}c^1_ka^1_{k}=\sum_{l=0}^{q-1}c^2_l a^2_{l}=1,\label{prime}
\eeq
i.e.  $\{a^j_{i}\}$ are co-prime integers for each $j=1,2$.

Then by repeatedly using \eqref{200.100}-\eqref{200.200} we arrive at
\beqn
h(\bn+(1,0))&=&h\lt(\bn+(\sum_{k}c^1_k a^1_{k},0)\rt)\nn
= h(\bn)+\im r_1 \mod \im 2\pi\\
h(\bn+(0,1))&=&h\lt(\bn+(0,\sum_{l}c^2_l a^2_{l})\rt)
=h(\bn)+\im r_2 \mod \im 2\pi
\eeqn
where
\beq
\label{rr}
r_1=\sum_{k=0}^{q-1}c^1_k\Delta^1_k, \quad r_2=\sum_{l=0}^{q-1}c^2_l\Delta^2_l.
\eeq
Therefore, we obtain \eqref{101} with $\br=(r_1,r_2)$ given by \eqref{rr}. 
Following through the rest of argument we can prove the following result.
\medskip
\begin{theorem}  \label{thm:mix}Suppose $f$ does not vanish in $\IZ_n^2$. For the perturbed raster scan \eqref{perturb}, let $\{\delta^i_{j_k}\}$ be the subset of  perturbations  satisfying 
 \beq
\label{small1'}
\tau&\ge&\max_{i=1,2}\{|a^i_{j_k}|+\delta_{j_k+1}^i-\delta_{j_k}^i\}\\
\label{cover2'}
2\tau&\le& m-\max_{i=1,2}\{\delta^i_{j_k+2}-\delta^i_{j_k}\}\\
\label{small2'}
m-\tau&\ge &1+\max_{k'}\max_{i=1,2}\{|a^i_{j_k}| +\delta^i_{k'+1}-\delta^i_{k'}\} 
\eeq
where $a^i_j=2\delta^i_{j+1}-\delta^i_j-\delta^i_{j+2}.$
Suppose 
\beq
\label{coprime}
\gcd_{j_k}\lt(|a^i_{j_k}|\rt)=1,\quad  i=1,2.
\eeq
 Let $\br=(r_1,r_2)\in \IR^2$ be given by \eqref{rr} and 
$\{c^j_i\}$ be any solution to \eqref{prime} such that $\{c^i_{j_k}\}$ are the only nonzero entries. 

Then both the object and probe errors have  a constant scaling factor and an affine phase profile:
\beq
\label{101'}
 g(\bn)/f(\bn)&=&\alpha^{-1}(0)\exp(\im \bn\cdot \br),\\
\nu^0(\bn)/\mu^0(\bn)&=&\alpha(0)
\exp(\im\phi(0)- \im \bn\cdot \br).\label{phase-error}\label{200.11}
\eeq
Further the block phases have an affine profile:
\beq\label{under-determined}
{\theta_{kl}=\theta_{00}}
+\bt_{kl}\cdot\br \mod  2\pi,
\eeq
for $k,l=0,\cdots,q-1.$

\end{theorem}
\medskip
\begin{rmk} 
It can be verified through a tedious calculation that \eqref{under-determined} (with   \eqref{D1}-\eqref{D2}, \eqref{prime} and \eqref{rr})  is an underdetermined linear
system for $\{\theta_{kl}\}$, which is consistent with the fact that the affine phase ambiguity \eqref{lp1}-\eqref{lp2} is inherent to
any blind ptychography.

\end{rmk}
\begin{proof}
\commentout{
To complete the proof of \eqref{101'}, it remains to show 
 that $D^1_{kl}$ and $D^2_{kl}$ are both non-empty. 
Indeed, under \eqref{small},  we can verify that
\beqn
\nn E_1&=&[\cM^{kl}-(a^1_{k},0)+(\tau+\delta^1_{k+1}-\delta^1_k,0)]\\
\nn &&\cap [\cM^{kl}-(a^1_{k},0)+2(\tau+\delta^1_{k+1}-\delta^1_{k},0)]\\
\nn E_2&=&[\cM^{kl}-(0,a^2_{l})+(0,\tau+\delta^2_{l+1}-\delta^2_l)]\\
\nn &&\cap [\cM^{kl}-(0,a^2_l)+2(0,\tau+\delta^2_{l+1}-\delta^2_{l})]
\eeqn
are nonempty sets
\commentout{
Note that for $j=1,2$, and $k=0,\dots,q-1$, 
\[
a_k^j-(\tau+\delta^j_{k+1}-\delta^j_k)=-\tau-\delta^j_{k+2}+\delta^j_{k+1} 
\]
}
since the differential shift between the two blocks in each $E_i, i=1,2$  is less than 
\commentout{
\[
a^1_k \quad \mbox{or}\quad \tau+\delta^1_{k+1}-\delta^j_k \quad \mbox{or}\quad \tau+\delta^1_{k+2}-\delta^1_{k+1}
\]
 and hence less than $m$ by \eqref{small}. 
 }
Therefore $D^1_{kl}\neq \emptyset, D^2_{kl}\neq \emptyset$. 
}

It remains to verify \eqref{200.11} and \eqref{phase-error} which  follow immediately from  \eqref{200.1} and \eqref{101'}. 

The block phase relation \eqref{under-determined} follows upon
substituting $\bt=\bt_{kl}$ and \eqref{101'} into  \eqref{200.1}.  

To summarize, we have shown that the scaling factor in \eqref{200.11} and the affine phase ambiguity, in \eqref{101'} and \eqref{phase-error}, are
the only ambiguities for the slightly perturbed raster scan \eqref{perturb}.

\end{proof}

\section{Numerical experiments}\label{sec:num}
In this section we  demonstrate geometric convergence  for blind ptychography with
the perturbed raster scan \eqref{perturb}.

\begin{figure}[t]
\centering
\subfigure[$f$'s real part]{\includegraphics[width=4cm]{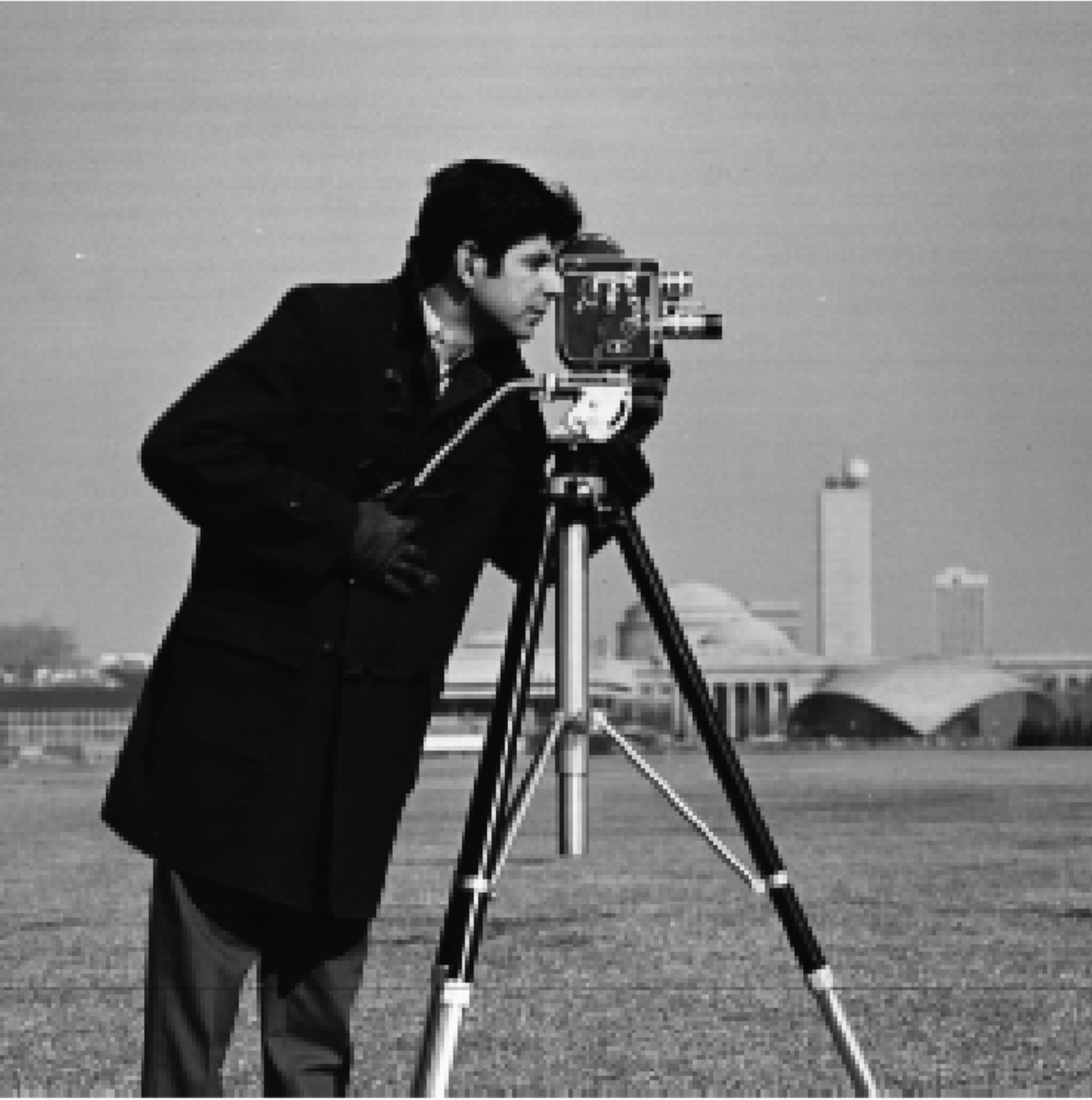}}\quad
\subfigure[$f$'s imaginary part]{\includegraphics[width=4cm]{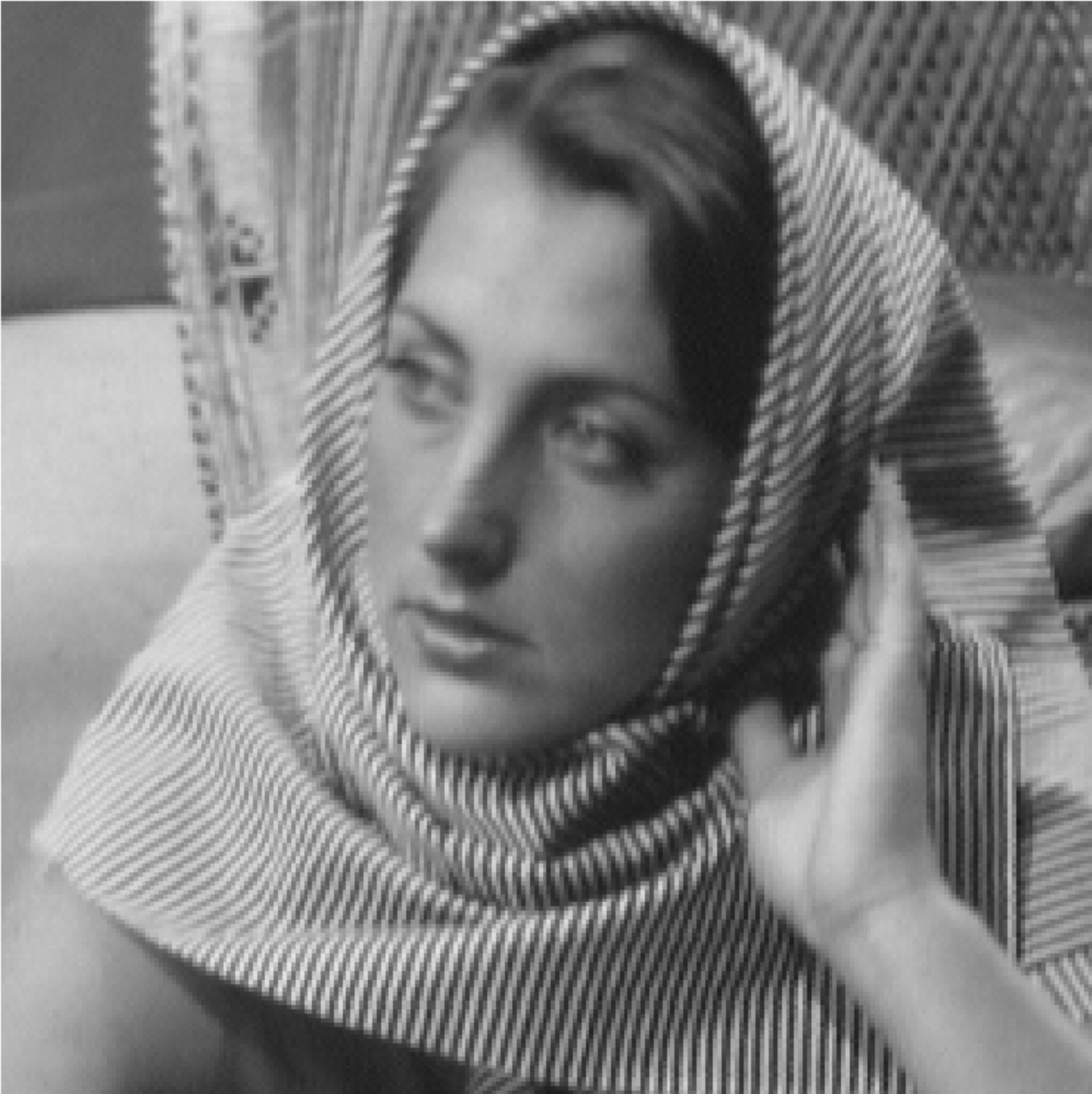}}\quad
\subfigure[Randomly phased probe]{\includegraphics[width=5.1cm]{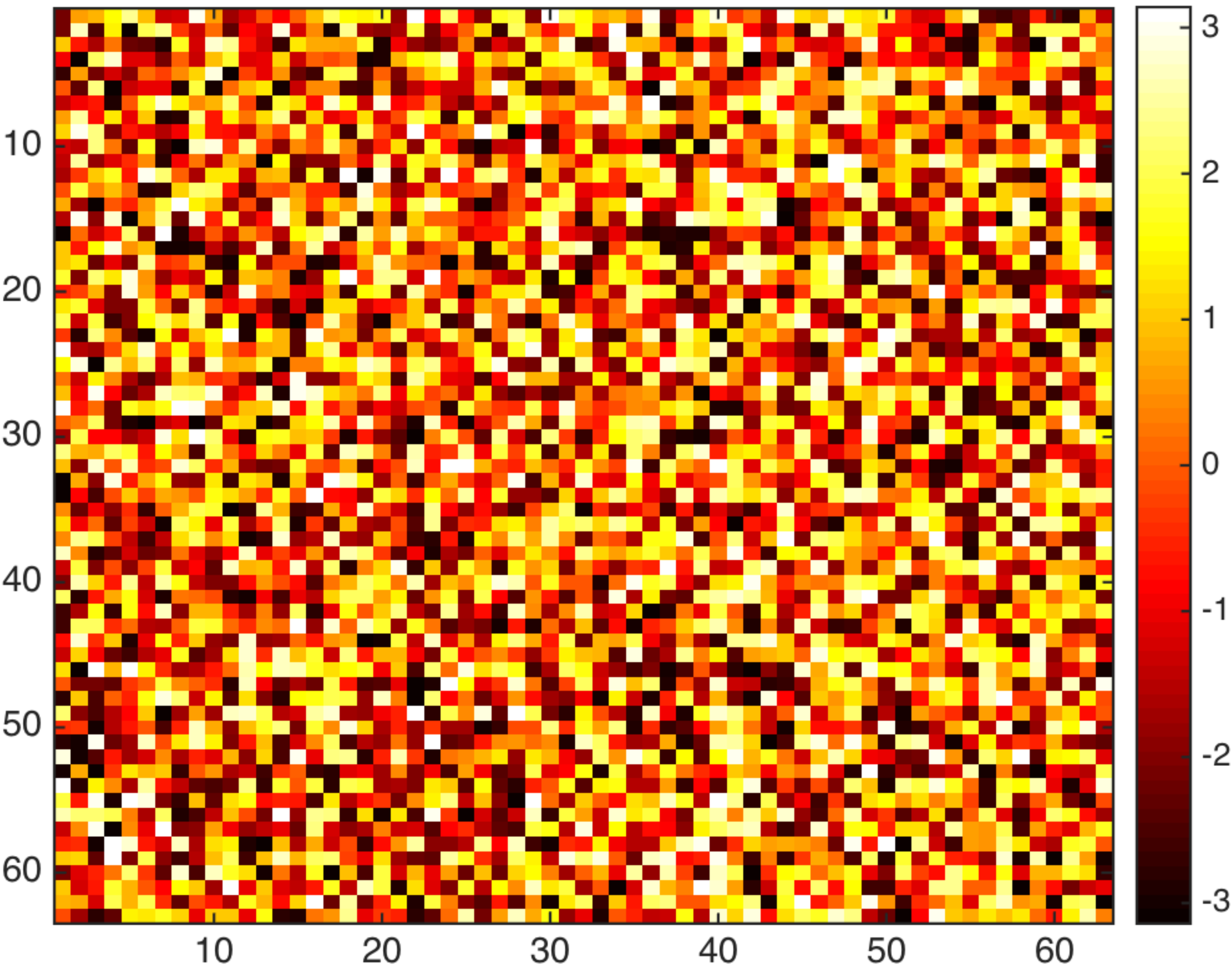}}
  \caption{The real part (a) and the imaginary part (b) of the object  and (c) randomly phased probe $\mu^{00}$.}
  \label{fig2}
\end{figure}%

\begin{figure}
\subfigure[RE with \eqref{perturb}]{\includegraphics[width=7cm]{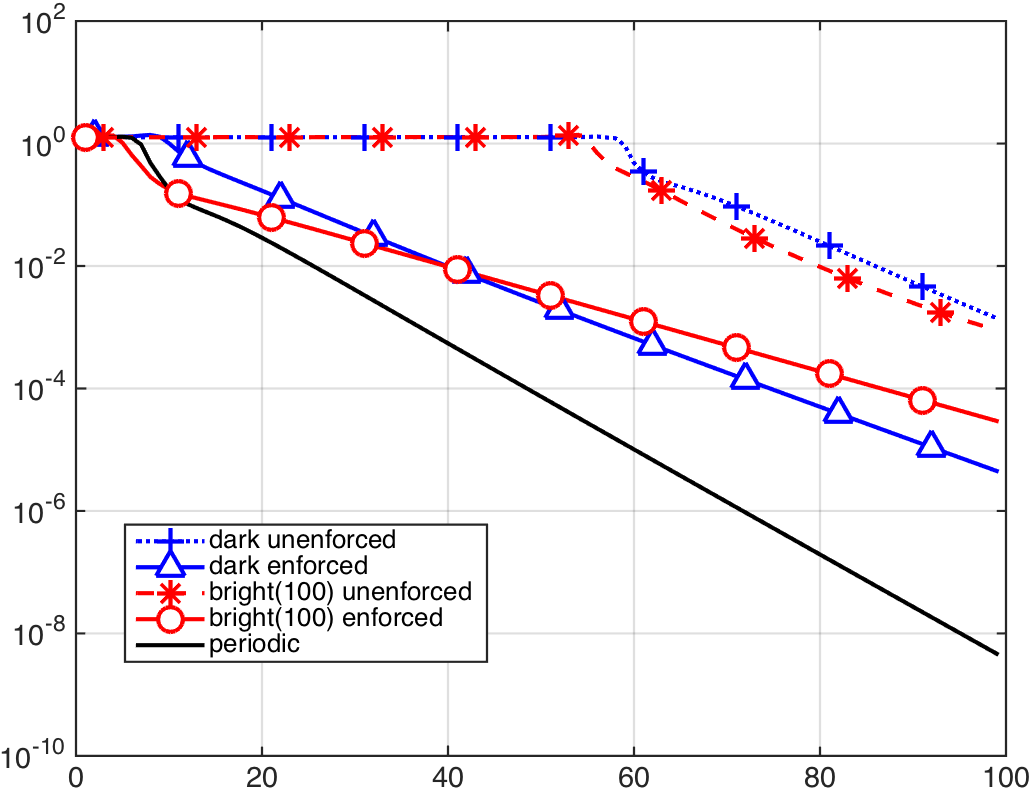}}\hspace{1cm}
\subfigure[RE with \eqref{perturb2}]{\includegraphics[width=7cm]{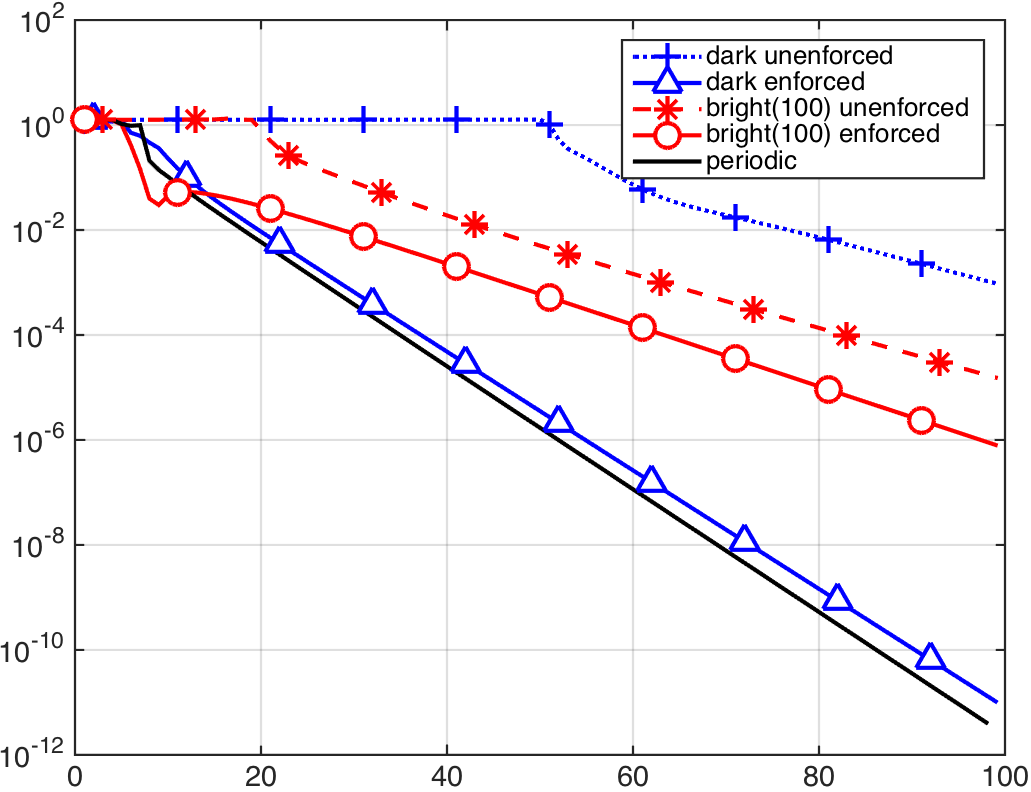}}
\caption{RE for  various boundary conditions with the sampling scheme (a) \eqref{perturb} and (b) \eqref{perturb2}.}
\label{fig3}
\end{figure}

Let $\cF(\nu,g)\in \IC^N$ be the totality of the Fourier (magnitude and phase) data  corresponding to the probe $\nu$ and the object $g$ such that $|\cF(\mu, f)|=b$ where $b$ is the noiseless ptychographic data. Since $\cF(\cdot,\cdot)$ is a bilinear function, $A_k h:=\cF(\mu_k, h), \, k\geq 1, $ defines a matrix $A_k$  for the $k$-th probe estimate $\mu_k$ and $B_k\eta :=\cF(\eta, f_{k+1}), \ k\geq 1,$  for  the $(k+1)$-st image estimate $f_{k+1}$ such that $A_k f_{j+1}=B_j\mu_k$, $j\geq 1, k\geq 1$. Let $P_k=A_kA_k^\dagger $ be the orthogonal projection onto the range of $A_k$ and
$R_k=2P_k-I$  the corresponding reflector.  
Likewise, let $Q_k=B_kB_k^\dagger $ be the orthogonal projection onto the range of $B_k$ and $S_k$ the corresponding reflector.

\begin{algorithm}
\caption{Alternating minimization (AM)}\label{alg: suedo algorithm}
\begin{algorithmic}[1]
\State Input: initial probe guess $\mathbf{\mu}_1$. 
\State Update the object estimate  
$\quad
f_{k+1}=\arg\min\cL(A_kg)$ s.t. $g\in \IC^{n\times n}$.
\State  Update the probe estimate  
$\quad
\mu_{k+1}=\arg\min\cL(B_k\nu)$ s.t.
$\nu\in \IC^{m\times m}$. 
\State Terminate if $\||B_k\mu_{k+1}|- b\|_2 $ stagnates or is less than tolerance; otherwise, go back to step 2 with $k\rightarrow k+1.$
\end{algorithmic}
\end{algorithm}

We use the objective function  
\[
\cL(y)=\half \| |y|- b\|_2^2
\]
and a randomly chosen initial probe guess satisfying
\[
\Re\lt[\overline{\mathbf{\mu}_1}(\bn)\odot \mu^{00}(\bn)\rt]>0,\quad \forall \bn,
\]
i.e. each pixel of the probe guess is aligned with the corresponding pixel of the true probe positively. 
The inner loops for updating the object and probe estimates are carried out by the Douglas-Rachford splitting method
as detailed in \cite{DRS-ptych,ptych-unique}: At epoch $k$, for $l=1,2,3,\dots$
\beqn
u_{k}^{l+1} &= &\frac{1}{2}u_{k}^{l} +\frac{1}{2}b\odot \sgn\big(R_ku_{k}^{l}\big),\quad u_k^1=u_{k-1}^\infty\\
v_{k}^{l+1}& =& \frac{1}{2} v_{k}^{l}+\frac{1}{2}b\odot \sgn{\Big(S_k v_{k}^{l}\Big)} ,\quad v_k^1=v_{k-1}^\infty
\eeqn
with the object estimate $f_{k+1}= A_k^\dagger u_{k}^{\infty}$ and
the  probe estimate $ \mu_{k+1}=B_k^\dagger v_{k}^{\infty}$ where $u^\infty_k$ and $v_k^\infty$ are terminal values
of the $k$-th epoch of the inner loops.
In the simulation for Fig. \ref{fig3} we keep the maximum number of iterations in
the inner loop to 30. 

To discount the constant amplitude offset and the linear phase ambiguity  we consider the following relative error (RE) 
 for the recovered image $f_k$ and probe $\mu_k$ at the $k^{th}$ epoch:
\beq\label{RE}
\mbox{RE}(k)&= &\min_{\alpha\in \mathbb{C}, \mathbf{k} \in \mathbb{R}^2}\frac{\|f(\mathbf{k}) - \alpha e^{-\imath {2\pi}\mathbf{k}\cdot \mathbf{r}/n} f_k(\mathbf{k})\|_2}{\|f\|_2}
\eeq

The image  is 256-by-256 Cameraman+ $\im$ Barbara (CiB). 
We use the randomly phased probe $\mu^{00}(\bn)=\exp[\im \phi(\bn)]$ where $[\phi(\bn)]$ are  $60\times 60$ i.i.d. uniform random variables over $[0,2\pi)$.  
We let  $\delta_{k}^1$ (resp. $\delta_{kl}^1$) and $\delta^2_l$  (resp. $\delta_{kl}^2$)  to be {i.i.d. uniform random variables  over $\lb -4,4\rb$}. 
In other words, the adjacent probes overlap by an average of $1-\tau/m=50\%$.

When the probe steps outside of the boundary of the object domain, the 
area $\cM\setminus \IZ_n^2$ 
needs special treatment in the reconstruction process. 

The periodic boundary condition forces 
the slope $\br$ in the linear phase ambiguity to be integers.
The dark-field and  bright-field  boundary conditions assume zero and nonzero ($=100$ in the simulation) values, respectively, in $\cM\setminus \IZ_n^2$. 
When the bright-field boundary condition is present in the simulation data and enforced in reconstruction,  the linear phase ambiguity disappears from the object estimate. On the contrary, enforcing the dark-field boundary condition can not remove
the linear phase ambiguity. In both cases, however, enforcement of boundary condition in reconstruction speeds
up the convergence as shown in Figure \ref{fig3}.

Figure \ref{fig3} shows that the sampling scheme \eqref{perturb2} generally outperforms \eqref{perturb} with a faster convergence rate, indicating that higher level of disorder
in the grid pattern is better for blind ptychohgraphy.

\section{Conclusions}\label{sec:con}
\commentout{
\begin{figure}
\centering
\subfigure[concentric circles]{\includegraphics[width=6cm]{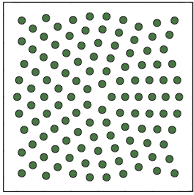}}\hspace{1cm}
\subfigure[Fermat spirals]{\includegraphics[width=6cm]{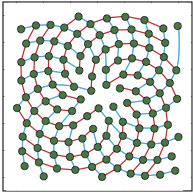}}
\end{figure}
}
We have studied the artifacts in blind ptychographic reconstruction from the perspective of
uniqueness theory of inverse problems and identified the periodic ambiguities in the raster scan
ptychography as the raster grid pathology reported in the optics literature. 

We have given a complete characterization of blind ptychographic ambiguities for the raster scan
including the periodic and non-periodic ambiguities. The non-periodic ambiguity have
an affine profile mirroring that of the block phases. To the best of our knowledge, such an ambiguity
has not been reported in the literature. 

We have presented a slightly perturbed under-shifted raster scan and proved that such a scheme can remove all
the ambiguities except for those inherent to any blind ptychography, namely the scaling factor
and the affine phase ambiguity.  In comparison, the same goal is approached in \cite{STFT} not by changing
the raster scan but by considering only a set of generic objects. 

\commentout{Note that the affine phase ambiguity is of a different nature from
that of the aforementioned non-periodic ambiguity as their affine profiles manifest on different scales:
the former on the individual pixel level while the latter on the block scale. 
}

For the perturbed under-shifted raster scan \eqref{perturb} with small random $\delta^i_j$, it is highly
probable that the co-prime condition \eqref{coprime} holds for large $q$ and hence
only the scaling factor and the affine phase ambiguity are present under \eqref{small1}-\eqref{small2} \cite{supres-PIE}.  
It would be interesting to see if the analysis presented in Section \ref{sec:mix} can be extended to other scan patterns  in practice such as the concentric circles \cite{circle,TM13,DM08}, the Fermat spiral \cite{optimal}
and those designed for Fourier ptychography \cite{optimal}.

In a noisy ptychographic experiment with the raster scan, as the step size shrinks, raster grid pathology becomes
less apparent and eventually invisible before the step size reaches 1 \cite{artifact} (cf. Corollary \ref{cor:raster}).
The affine phase ambiguity and the raster grid pathology can also be suppressed by additional prior information
such as the Fourier intensities of the probe \cite{March16}.

 \commentout{
\section{Absence of phase drift}
\label{sec:u2}

In this section, we discuss a way to eliminate the affine phase ambiguity with the help of the periodic boundary condition. This is accomplished by the imposing the size of the random probe
and the fact that for any cyclic group $S$ the block phases form an arithmetic progression (Theorem \ref{lem2}). 

\begin{theorem} \label{thm:u2} Consider any cyclic subgroup $S$ of order $s$. Suppose that 
\beq
\label{u2}
n\le 2(m-\tau),\quad \mbox{(equivalently,  $q \le 2(p-1)$)}.
\eeq
Then 
\beq
\label{40}
\nu^{k}\odot g^k=e^{\im \theta}\mu^{k}\odot f^k,\quad\forall k=0,\dots,s-1,
\eeq
for some constant $\theta\in \IR$ independent of $k$. 

\end{theorem}
\begin{proof}

By \eqref{200.4'}-\eqref{200.4}, 
\beqn
e^{\im \theta_k}\nu^k(\bn)/\mu^k(\bn)=e^{\im \theta_{k+1}}\nu^{k+1}(\bn)/\mu^{k+1}(\bn),\quad
\forall \bn\in \cM^k\cap \cM^{k+1}\cap \supp(f)
\eeqn
and by Theorem \ref{lem2}
\[
 \theta_k=k\Delta \theta,\quad k=0,\dots,s-1
 \]
where $\Delta \theta$ times $s$ is an integer multiple of $2\pi$. 

If $s\Delta \theta$ is an integer multiple of $2\pi$, then the phase of 
$\nu^{00}(\bn)/\mu^{00}(\bn)$ 
would grow out of the probe phase constraint  under the constraint $q\delta \le 2(p-1)$.  

\end{proof}
}

\commentout{
\section{Ambiguities with loose support}\label{sec:loose}

Objects of loose support can create ptychographic ambiguities other than
a constant scaling factor and an affine phase factor. Further, 
 \eqref{7.1} may be altogether untenable for loose support in general. 
Although the following examples involve ptychographic measurements with only two diffraction patterns,
more elaborate examples can be constructed and omitted for simplicity of presentation.
The important feature of the examples is that besides loose support the object values are arbitrary.

\begin{ex}\label{ex3.1}
For $m=3n/5$ consider $\cT=\{(0,0),\bt\}$ where  $\bt=(2m/3, 0)$. Evenly partition $f^{0}$ and $f^\bt$ into three  $m\times m/3$ parts  as 
$f^0=[0,f^0_{1}, f^0_{2}]$ and $f^\bt=[f^1_{0},  f^1_{1}, 0]$ 
with the overlap $f^{0}_{2}=f^{1}_{0}$. Likewise, partition the probe as
$\mu^0=[\mu^0_0,\mu^0_1,\mu^0_2], \mu^\bt=[\mu^1_0,\mu^1_1,\mu^1_2]$ where $\mu^\bt$ is just the $\bt$-shift of $\mu^0$,
i.e. $\mu^\bt(\bn+\bt)=\mu^0(\bn)$. 

Let $
\nu^0=\mu^0, \nu^\bt=\mu^\bt$
 and $g^0=[g^0_{0}, g^0_{1}, 0], g^\bt=[0, g^1_{1}, g^1_{2}]$  
where
\beq
\label{11.1}
g^0_{0}= f^0_{1}\odot \mu^0_1/\mu^0_0,&&g^0_{1}=f^0_{2}\odot \mu^0_2/\mu^0_1\\
\label{11.2} g^1_{1}= f^1_{0}\odot \mu^1_0/\mu^1_1,&&g^1_{2}=f^0_{1}\odot \mu^1_1/\mu^1_2.
\eeq
Clearly, $g=[g^0_{0}, g^0_{1}, 0, g^1_{1}, g^1_{2}]$ is different from $f=[0,f^0_{1}, f^0_{2},  f^1_{1}, 0]$.

It is straightforward to check that for $\mbm=(m/3,0)$
\beq
\label{3.24}g^0(\bn)\nu^0(\bn)=f^0(\bn+\mbm) \mu^0 (\bn+\mbm),\quad\bn\in \cM^0\\
\label{3.25}g^\bt(\bn)\nu^\bt(\bn)=f^\bt(\bn-\mbm) \mu^\bt (\bn-\mbm),\quad\bn\in \cM^\bt
\eeq
and hence $g^0\odot\mu^0$ and $ g^\bt \odot\mu^\bt$ produce  the same  diffraction patterns
as $f^0\odot\mu^0$ and $f^\bt\odot\mu^\bt$. 

On the other hand, $g^0(\bn) \nu^0 (\bn)\neq e^{\im \theta_0}f^0(\bn)\mu^0(\bn)$ 
and $g^\bt(\bn) \nu^\bt (\bn)\neq e^{\im \theta_\bt} f^\bt(\bn) \mu^\bt (\bn)$ in general. 
Hence \eqref{7.1} is violated. 
\end{ex}

The ambiguity in Example \ref{ex3.1} is due to a loose object support which can also produce the ambiguity of conjugate inversion
as follows.

\begin{ex} \label{ex3.3} For $m=n$ consider $\cT=\{(0,0),\bt\}$ where $\bt=(m/2,0)$ with the periodic boundary condition.
Evenly partition $f^0$ and $f^\bt$ into two  $m\times m/2$ parts as 
$f^0=[0,f^0_{1}]$ and $f^\bt=[f^1_{0}, 0]$ 
with the overlap $f^0_{1}=f^1_{0}$. Likewise, partition the probe as
$\mu^0=[\mu^0_0,\mu^0_1], \mu^\bt=[\mu^1_0,\mu^1_1]$ where $\mu^\bt(\bn+\bt)=\mu^0(\bn)$. 

Let $
\nu^0=\mu^0, \nu^\bt=\mu^\bt$
 and $g^0=[g^0_{0}, 0], g^\bt=[0, g^1_{1}]$  
 where
\beq
\label{3.26'}
g^0_{0}(\bn)&= &\bar f^0_{1}(\bN-\bn)\bar\mu_1^0(\bN-\bn)/\mu^0_0(\bn)\\
\label{3.27'}g^1_{1}(\bn)&= &\bar f^1_{0}(\bN-\bn)\bar \mu_0^1(\bN-\bn)/\mu_1^1(\bn). 
\eeq
In other words, $g=[g^0_0,0,g^1_1]$ is the twin-like image of $f=[0,f^1_1,0]$. Since $2\bt=(m,0)=(0,0)$ under the periodic boundary condition, we have $g^0_{0}=g^1_{1}$.  So the construction \eqref{3.26'}-\eqref{3.27'}  is consistent with the periodic boundary condition. 

Setting $\nu^0=\mu^0$ we have
\beqn
g^0(\bn)\nu^0(\bn)&=&\bar f^0(\bN-\bn) \bar \mu^0 (\bN-\bn)\\
g^\bt(\bn)\nu^\bt(\bn)&=&\bar f^\bt(\bN-\bn) \bar\mu^\bt (\bN-\bn)
\eeqn
and hence $g^0\odot\nu^0$ and $ g^\bt \odot\nu^\bt$ produce  the same  diffraction patterns
as $f^0\odot\mu^0$ and $f^\bt\odot\mu^\bt$. 

On the other hand, if $f^0=[f^0_0, f^0_1]$  has a tight support in $\cM^0$, then  \eqref{3.26'}-\eqref{3.27'} is an inconsistent construction of object estimate 
since $g^0=[g^0_0, g^0_1]$ and $ g^\bt=[0, g^1_1]$ violate the overlap condition for $g^0_1\neq 0$. 

\end{ex}

It is important to note that if just one part $f^\bs$ has a tight support in $\cM^\bs$ and if
the set of parts $\{f^\bt: \bt\in \cT\}$ is strongly connected, then the ambiguities associated with loose support
disappear with an overwhelming  probability. 
}


\section*{Acknowledgment}
{This work was supported by the National Science Foundation under Grant  DMS-1413373.}
I thank National Center for
Theoretical Sciences (NCTS), Taiwan,   where  the present work was completed, for the hospitality  during my visits  in June and August 2018. 
I am grateful for Zhiqing Zhang for preparing Fig. \ref{fig3}.  \\

\end{document}